\documentclass{article}

\usepackage[utf8]{inputenc} 
\usepackage[T1]{fontenc}    
\usepackage{lmodern}
\usepackage{hyperref}       
\usepackage{url}            
\usepackage{booktabs}       
\usepackage{amsfonts}       
\usepackage{nicefrac}       
\usepackage{microtype}      

\usepackage[pdftex]{graphicx}
\usepackage{subfigure}
\usepackage{algorithm}
\usepackage{algorithmic}
\usepackage{amsmath, amsthm, amssymb}
\usepackage[margin=1.25in]{geometry}

\newcommand{\betahw}{\hat\beta_W}
\newcommand{\betahols}{\hat\beta_{OLS}}
\newcommand{\sigmah}{\hat\sigma^2}
\newcommand{\tran}{\mathsf{T}}
\newcommand{\simiid}{\stackrel{\mathrm{iid}}{\sim}}
\newcommand{\dnorm}{\mathcal{N}}
\newcommand{\e}{\mathbb{E}}
\newcommand{\var}{\mathrm{Var}}

\newcommand{\diag}{\mathrm{diag}}

\newtheorem{model}{Model}
\newtheorem{lemma}{Lemma}
\newtheorem{theorem}{Theorem}

\title{Confidence Intervals for Algorithmic Leveraging in Linear Regression}

\author{Katelyn Gao\\
Stanford University \footnote{Now at Intel Labs}}

\date{}

\begin{document}

\maketitle

\begin{abstract}
The age of big data has produced data sets that are computationally expensive to analyze and store. Algorithmic leveraging proposes that we sample observations from the original data set to generate a representative data set and then perform analysis on the representative data set. In this paper, we present efficient algorithms for constructing finite sample confidence intervals for each algorithmic leveraging estimated regression coefficient, with asymptotic coverage guarantees. In simulations, we confirm empirically that the confidence intervals have the desired coverage probabilities, while bootstrap confidence intervals may not. 
\end{abstract}

\section{Introduction}\label{ref:intro}


A popular method to deal with modern massive data sets with large volume is to sample a representative data set that is much smaller than the original data set and then carry out analysis on the representative data set. The sampling of the representative data set is usually dependent on the analysis method employed. As long as the sampling process is not too computationally expensive, if the size of the sample is much smaller than the size of the original data set, we could have considerable savings on the total computation time.

Algorithmic leveraging refers to the special case where the probability an observation is sampled is positively correlated with the influence of each observation on the results of data analysis. The definition of influence varies based on the data analysis problem considered.  

We are concerned with algorithmic leveraging for the problem of least squares approximation (linear regression) \cite{DMM06}, but algorithmic leveraging has been applied to many other data analysis problems with large data sets. They include low-rank matrix approximation \cite{MD09} and least absolute deviations regression \cite{CDMMMW13}.

Previous literature has primarily been concerned with estimation, providing high probability bounds on errors. However, in many practical regression settings, we also would like to perform uncertainty quantification. Therefore, in this paper we explore how to efficiently construct confidence intervals and tests of significance for the estimated coefficients. We utilize the framework of Ma et al. \cite{MMY15}, which computes the expectation and variance of algorithmic leveraging estimates for linear regression.

Assume that the data come from the following model.
\begin{model}\label{eq:refmodel}
For $i=1,...,N$
\[y_i=x_i^\tran\beta+\epsilon_i, \quad\epsilon_i\simiid\dnorm(0,\sigma^2)\]
where $x_i\in\mathbb{R}^p$ are the $p$ predictors and $y_i \in \mathbb{R}$ are the responses. 
\end{model}
The statistically efficient estimator of the regression coefficient vector $\beta$ is the ordinary least squares (OLS) estimator $\betahols$, which requires $O(Np^2)$ time to compute (see Section~\ref{sec:ols}). When both $N$ and $p$ are large, that would be expensive. This is why a data reduction method like algorithmic leveraging is desirable. 

In section~\ref{sec:background}, we introduce ordinary least squares and algorithmic leveraging for linear regression, and present some useful lemmas. Section~\ref{sec:inference} shows how to construct confidence intervals and tests of significance using those estimated coefficients in $o(Np^2)$ time. The confidence intervals are exact when $\sigma^2$ is known and asymptotically exact when $\sigma^2$ is not known.
In section~\ref{sec:experiments}, using simulated data, we confirm that our proposed confidence intervals have at least the nominal coverage probability and that our tests of significance control the type 1 error rate and have low type 2 error rates. In contrast, the bootstrap, a popular approach to obtaining confidence intervals for complex estimators, is more computationally expensive and may lead to confidence intervals with too small coverage probabilities.
Section~\ref{sec:conclusion} concludes and discusses possible future work.

\section{Background and Preliminary Work}\label{sec:background}

In this section, we provide some background on least squares estimation of linear regression models and describe algorithmic leveraging as applied to that problem. We end with a characterization of the distribution of the resulting regression coefficient estimates.

Let $v_j$ indicate the $j$th entry of a vector $v$, and $(M)_{ij}$ indicate entry $(i,j)$ of a matrix $M$.

\subsection{Ordinary Least Squares and Statistical Leverage Scores}\label{sec:ols}

Suppose that we have data from Model~\eqref{eq:refmodel}. Let $X$ be the $N\times p$ matrix with rows $x_i^\tran$ and $Y$ be the $N$-dimensional vector of $y_i$'s. The OLS problem is
\begin{align*}
\arg\min_{\beta\in\mathbb{R}^p}\|Y-X\beta\|_2^2
\end{align*}
and its solution is 
\begin{align}\label{eq:betaols}
\betahols &= (X^\tran X)^{-1}X^\tran Y
\end{align}
with
\begin{align*}
\var(\betahols) &= \sigma^2(X^\tran X)^{-1}.
\end{align*}

By the Gauss-Markov Theorem \cite{W05}, $\betahols$ is the minimum variance unbiased estimator of $\beta$. By using the singular value decomposition (SVD) of $X$, $\betahols$ and, when $\sigma^2$ is known, $\var(\betahols)$ can be computed in $O(Np^2)$ time \cite{GV96}. However, when both $N$ and $p$ are large, doing so is expensive.

Let the OLS fitted values be
\begin{align}\label{eq:hatmatrix}
\hat Y &= X(X^\tran X)^{-1}X^\tran Y = H Y
\end{align}
where $H=X(X^\tran X)^{-1}X^\tran$ is called the hat matrix. 

The statistical leverage score of the $i$th observation is defined to be the $i$th diagonal entry of $H$, $h_{ii}=x_i^\tran(X^\tran X)^{-1}x_i$. The statistical leverage scores may be directly computed from the SVD of $X$, and thus requires $O(Np^2)$ computation time.
$h_{ii}$ is considered to be the influence of the $i$th observation on the regression results \cite{CH86}. For linear regression, algorithmic leveraging samples observations so that those with higher statistical leverage scores are more likely to be sampled and solves a least squares problem on the sample.

\subsection{Algorithmic Leveraging in Linear Regression}\label{sec:alglev}

This section summarizes the traditional algorithmic leveraging procedure for linear regression. Given data from Model~\eqref{eq:refmodel} and a distribution $\pi$ over the observations, we execute Algorithm~\ref{alg:alglev}.

\begin{algorithm}[ht]
	\caption{Algorithmic Leveraging \cite{MMY15}}
	\label{alg:alglev}
	\begin{algorithmic}[1]
	\STATE Sample with replacement $r$ observations according to the probabilities $\pi_i$ for $i=1,...,N$.
	\STATE Scale each sampled observation by $1/\sqrt{r\pi_i}$. 
	\STATE Solve the OLS problem on the scaled sampled observations, and return the estimate of $\beta$.
	\end{algorithmic}
\end{algorithm}
Note that in theory $r$ could be any natural number, but in practice we would choose $r<N$.

Step $1$ of Algorithm \ref{alg:alglev} has computational complexity $O(r)$, and step $2$ $O(rp)$. Step $3$, as in section \ref{sec:ols}, requires $O(rp^2)$ computation time. Therefore, as long as the computation time to obtain the distribution $\pi$ is $o(Np^2)$, the total computation time of Algorithm~\ref{alg:alglev} is $o(Np^2)$. 

The ideal distribution $\pi_i \propto h_{ii}$; by doing so, we are more likely to sample observations that are influential. Because computing the statistical leverage scores requires $O(Np^2)$ time, in practice, approximations such as those of Clarkson et al. \cite{CDMMMW13} and Drineas et al. \cite{DMMW12} are used; they are computable in $O(Np\log p)$ time and $O(Np\log N)$ time, respectively.

However, for our theoretical results we do not assume that $\pi$ was computed using either of the above algorithms. 
We only assume that $\pi$ is not dependent on the responses and that $\pi$ does not have any entries equal to zero. 

Let $S_X$ be the $N \times r$ matrix where entry $(j,i)$ is $1$ if the $i$th sample is the $j$th observation in the original data set and $0$ otherwise. The sampled data can be written as $(S_X^\tran X,S_X^\tran Y)$. Let $D$ be the $r\times r$ diagonal matrix with $i$th diagonal element $1/\sqrt{r\pi_j}$ if the $i$th sample is the $j$th observation of the original data set.

Then, step $3$ of Algorithm \ref{alg:alglev} can be written as
\begin{align}\label{eq:wls}
\arg\min_{\beta\in\mathbb{R}^p}\|DS_X^\tran Y-DS_X^\tran X\beta\|_2^2
\end{align}
with solution
\begin{align}\label{eq:betaw}
\begin{split}
\betahw &= (X^\tran W X)^{-1} X^\tran W Y \\
W &= S_X D^2 S_X^\tran.
\end{split}
\end{align}
For convenience, we call the result of Algorithm~\ref{alg:alglev}, $\betahw$, the algorithmic leveraging estimate.

\subsection{Asymptotic Normality of the Algorithmic Leveraging Estimates}\label{sec:asymp}

In this section, we study the distribution of $\betahw$. We first note that $W$, as defined in $\eqref{eq:betaw}$, is a $N \times N$ diagonal matrix. Therefore, we can write $W=\diag(\mathbf{w})$, where $\mathbf{w}=(w_1,...,w_N)^\tran$. 

\begin{lemma}\label{lemma:wasconv}
For given $N$, as $r \to \infty$, $\mathbf{w}$ converges almost surely to $\mathbf{1}_N$.
\end{lemma}
\begin{proof}
See Section \ref{sec:proofs:wasconv}. 
\end{proof}

Using the multivariate Lindeberg-Feller central limit theorem \cite{greene} and results from Ma et al. \cite{MMY15}, we obtain
\begin{theorem}\label{thm:betaclt}
Suppose that as $N\to\infty$, $X^\tran X/N$ converges to a finite positive definite matrix and $\lim_{N\to\infty}(\sum_{i=1}^N x_i x_i^\tran)^{-1}x_i x_i^\tran = 0$. Then as $r,N\to\infty$, $\betahw$ is approximately distributed as
\begin{align*}
\dnorm(\beta,\sigma^2(X^\tran X)^{-1}+\dfrac{\sigma^2}{r}(X^\tran X)^{-1}X^\tran \diag\left(\dfrac{(1-h_{ii})^2}{\pi_i}\right)X(X^\tran X)^{-1}).
\end{align*}
\end{theorem}
\begin{proof}
See Section \ref{sec:proofs:betaclt}
\end{proof}

This theorem can be used to construct confidence intervals for each element of $\beta$ that have the correct coverage probability as $r$ and $N$ approach infinity. However, doing so requires computing $(X^\tran X)^{-1}$, which requires $O(Np^2)$ computation time, and thus we proceed in a different direction.

\subsection{Bootstrap}\label{sec:bootstrap}

Because the main bottleneck in applying Theorem \ref{thm:betaclt} to construct confidence intervals is in computing the variance of $\betahw$, we could consider using the bootstrap as follows.
For a total of $B$ times, the bootstrap samples with replacement $N$ observations from the data set and applies Algorithm \ref{alg:alglev} to the sample. Then, for $j=1,\ldots,p$, the standard deviation $\delta_j$ of the $j$th coordinate of the algorithmic leveraging estimates is computed; the level $1-\alpha$ bootstrap confidence interval for $\beta_j$ is constructed as
\begin{align*}
(\betahw)_j \pm z_{1-\alpha/2}\delta_j
\end{align*}
where $z_{1-\alpha/2}$ is the $1-\alpha/2$ quantile of the standard normal distribution.
The corresponding test of significance would be to reject the null hypothesis $\beta_j=0$ when $|(\betahw)_j| > z_{1-\alpha/2}\delta_j$.

However, since in practice we must use fairly large values of $B$, usually at least on the order of $10^2$, the computational complexity of the bootstrap procedure is at least $O(BN)$. It is more expensive than our proposal presented in the next section, and we show that experimentally it may not lead to valid confidence intervals. Moreover, the asymptotic guarantees of the bootstrap would require both $N$ and $r$ to approach infinity, while our proposal's guarantees only require $r$ to approach infinity.

\section{Inference on \texorpdfstring{$\betahw$}{}}\label{sec:inference}

This section presents our proposed approach for uncertainty quantification for the algorithmic leveraging estimator $\betahw$, both when the error variance $\sigma^2$ is known and when it is unknown.

Consider the distribution of $\betahw$ conditional on $W$, i.e. conditional on the sample.
$\betahw$ is normally distributed with mean and variance
\begin{align}
&\e(\betahw \mid W) = \beta \label{eq:betahwexp} \\
\begin{split}\label{eq:betahwvar}
\var(\betahw \mid W) = \sigma^2(X^\tran WX)^{-1}X^\tran W^2X(X^\tran WX)^{-1}
\end{split}
\end{align}

We write
\begin{align*}
X^\tran W^2 X &= (S_X^\tran X)^\tran D^2S_X^\tran S_X D^2(S_X^\tran X)
\end{align*}
Entry $(i,j)$ of $S_X^\tran S_X$ is $1$ if the $i$th and $j$th sample are the same observation, and $0$ otherwise. Thus, $S_X^\tran S_X$ can be computed in $O(r^2)$ time. Note that we already know $S_X^\tran X$ and $D^2$ from computing $\betahw$. Since $S_X^\tran X$ is $r\times p$, $D^2$ is diagonal of size $r$, and $S_X^\tran S_X$ is $r \times r$, $X^\tran W^2 X$ can be computed in $O(pr^2)$ time. Because $(X^\tran W X)^{-1}$ can be found from the computation of $\betahw$, $\var(\betahw \mid W)$ can be computed in $O(pr^2)$ time. If we choose $r \sim N^{1/2-\delta}$ for some $\delta\geq 0$, then that computation time is $O(Np)$.

\subsection{\texorpdfstring{$\sigma^2$}{} is Known}\label{sec:infsigmaknown}

By the argument in the previous paragraph, using the following theorem, we can get exact confidence intervals for each element of $\beta$ in $O(Np)$ time.

\begin{theorem}\label{thm:CIsigmaknown}
For $0<\alpha<1$, an exact level $1-\alpha$ confidence interval for $\beta_j$ based on $\betahw$ is
\begin{align*}
&(\betahw)_j \pm z_{1-\alpha/2}\sigma\sqrt{((X^\tran WX)^{-1}X^\tran W^2X(X^\tran WX)^{-1})_{jj}}.
\end{align*}
\end{theorem}
\begin{proof}
See Section \ref{sec:proofs:CIsigmaknown}
\end{proof}

By the correspondence between confidence intervals and hypothesis testing, in $o(Np^2)$ time we can also test for the significance of each regression coefficient estimated by algorithmic leveraging using, for $j=1,\ldots,p$, the hypothesis $H_0:\beta_j=0$. Specifically, a test with significance level $\alpha$ rejects $H_0$ when 
\begin{align}
\begin{split}
& |(\betahw)_j| \geq z_{1-\alpha/2}\sigma\sqrt{((X^\tran WX)^{-1}X^\tran W^2X(X^\tran WX)^{-1})_{jj}}.
\end{split}
\end{align}

\subsection{\texorpdfstring{$\sigma^2$}{} Is Unknown}\label{sec:infsigmaunknown}

In practice, the error variance is rarely known. In this section, we discuss this more realistic situation.

We estimate $\sigma^2$ analogously to OLS. Letting ${\hat Y_W=X\betahw=X\beta+H_W\epsilon}$ be the predicted values, define
\begin{align}\label{eq:sigmahat}
\begin{split}
\sigmah &= \dfrac{1}{N-p}\|Y-\hat Y_W\|_2^2= \dfrac{1}{N-p}\|Y-H_WY\|_2^2=\dfrac{1}{N-p}\|(I-H_W)\epsilon\|_2^2,
\end{split}
\end{align}
where $H_W=X(X^\tran WX)^{-1}X^\tran W$ and $\hat Y_W=H_WY$. $\sigmah$ can be computed in $O(Np)$ time given $\betahw$. 

Classical statistical inference, which allows us to find exact confidence intervals for regression coefficients, computes the exact distribution of $\sigmah$. 
However, in our case, that distribution depends on the singular values of $I-H_W$, which in general cannot be computed in $o(Np^2)$ time. 
Instead, we utilize Lemma~\ref{lemma:wasconv}, which states that for $r$ large enough, with probability one $W\approx I_N$, where $I_N$ is the $N$-dimensional identity matrix. 

Intuitively, we may make the approximation that $W\approx I_N$. Then, $H_W \approx H$, the hat matrix from OLS, and our estimates $\betahw$ and $\sigmah$ approximate the regression coefficients and standard error estimate from OLS. 

Therefore, inspired by the classic confidence interval for OLS regression coefficients, we propose the following approximate level $1-\alpha$ confidence interval for $\beta_j$ based on $\betahw$:
\begin{align}\label{eq:CIsigmaunknown}
\begin{split}
&(\betahw)_j \pm t_{N-p,1-\alpha/2}\hat\sigma \sqrt{((X^\tran WX)^{-1}X^\tran W^2X(X^\tran WX)^{-1})_{jj}}
\end{split}
\end{align}
where $t_{N-p,1-\alpha/2}$ is the $1-\alpha/2$ quantile of $t_{N-p}$. 

These confidence intervals have the correct coverage probability asymptotically as $r$ approaches infinity.
\begin{theorem}\label{thm:CIsigmaunknown}
If $X$ has full rank, as $r\to\infty$, the confidence interval \eqref{eq:CIsigmaunknown} has coverage probability $1-\alpha$.
\end{theorem}
\begin{proof}
See Section \ref{sec:proofs:CIsigmaunknown}
\end{proof}

By the argument in the previous section and the fact that $\sigmah$ is computed in $O(Np)$ time, the above confidence interval can be computed in $o(Np^2)$ time. As before, in $o(Np^2)$ time we can also test for the significance of each regression coefficient estimated by algorithmic leveraging using, for $j=1,\ldots,p$, the hypothesis $H_0:\beta_j=0$. Specifically, we propose a test with approximate significance level $\alpha$ that rejects $H_0$ when 
\begin{align}\label{eq:sigtest}
\begin{split}
&|(\betahw)_j|\geq t_{1-\alpha/2,N-p}\hat\sigma \sqrt{((X^\tran WX)^{-1}X^\tran W^2X(X^\tran WX)^{-1})_{jj}}.
\end{split}
\end{align}

\section{Experimental Results}\label{sec:experiments}

In this section, we illustrate the behavior of our proposed confidence intervals \eqref{eq:CIsigmaunknown} and significance tests \eqref{eq:sigtest} and compare them to the bootstrap. Our experiments were carried out in $\textsf{R}$, using the special package \textit{mvtnorm} \cite{mvtnorm}. They follow the same setup as in Ma et al. \cite{MMY15}.

We use data simulated from model~\eqref{eq:refmodel}. We consider the data sizes $p=10,50$ and $N=1000,5000$. For each tuple $(p,N)$, we generate $x_1,\ldots,x_N$ independently from the multivariate $t$-distribution with three degrees of freedom and covariance matrix $\Sigma$, where $(\Sigma)_{i,j}=2\times 0.5^{|i-j|}$. This will give a matrix $X$ with some large statistical leverage scores and some small ones \cite{MMY15}. The sampling distribution $\pi$ was chosen to be proportional to the approximate leverage scores $\tilde{h}_{ii}$ computed using the algorithm in Drineas et al. \cite{DMMW12}.

For each $p$, we randomly choose half the entries of $\beta$ to be zero, a quarter to be $+1$, and the rest to be $-1$. Then, for each tuple $(X,\beta)$ we repeat the following procedure $100$ times:
\begin{enumerate}
\item Generate $Y \sim \dnorm(X\beta,9I_N)$
\item For $r=100,200,\ldots,N$,
\begin{itemize}
\item Carry out Algorithm~\ref{alg:alglev} with $\pi_i \propto \tilde{h}_{ii}$ to obtain $\betahw$.
\item Compute $\sigmah$ using \eqref{eq:sigmahat} and $(X^\tran WX)^{-1}X^\tran W^2X(X^\tran WX)^{-1}$ as described in Section~\ref{sec:infsigmaknown}. Record the computation time used.
\item Compute the bootstrap standard deviations $\delta_j$ as described in Section \ref{sec:bootstrap} with $B=100$ and record the computation time used.
\item For $\alpha=0.01,0.02,\ldots,1$, 
\begin{itemize}
\item For $j=1,\ldots,p$, compute a level $\alpha$ confidence interval for $\beta_j$ using \eqref{eq:CIsigmaunknown}. Calculate the actual coverage probability, i.e. the fraction of $j$'s for which the interval includes the true value of $\beta_j$; the nominal coverage probability is $\alpha$.
\item For $j=1,\ldots,p$, test $H_0: \beta_j=0$ at significance level $\alpha$ using \eqref{eq:sigtest}. Compute the proportion of type 1 and 2 errors.
\item Repeat the above two steps for the bootstrap, constructing the confidence intervals as described in Section \ref{sec:bootstrap} and using the corresponding tests of significance.
\end{itemize}
\end{itemize}
\end{enumerate}
Finally, for each tuple $(p,N,r,\alpha)$, for both our algorithm and bootstrap we compute the averages of the computation times, actual coverage probabilities, and proportions of type 1 and type 2 errors over the $100$ iterations.

For each tuple $(p,N)$, we plot four graphs to evaluate and compare the quality of our proposed confidence intervals and bootstrap confidence intervals. The first plots the computation time of our algorithm and that of bootstrap versus $r$. The second plots, for both our algorithm and the bootstrap, the actual coverage probability versus the nominal coverage probability of the confidence intervals at three small values of $r$. The third and fourth plot, for our algorithm and the bootstrap, the average type 1 error rate and average type 2 error rate of the significance tests at $\alpha=0.05$ versus $r$.

The graphs are shown in Figures \ref{fig:p10N1000}--\ref{fig:p50N5000}. In Section \ref{sec:experiments:roc}, for each pair of values of $N$ and $p$, we also present plots of the receiver operating characteristic (ROC) curve for various values of $r$ for both our algorithm and the bootstrap. We defer them to the appendix because they are usually used to evaluate classifiers and not methods for uncertainty quantification. 

We see that our algorithm is indeed faster than the bootstrap, with the gap widening as $r$ increases. 
Indeed, the rate of increase in computation time is greater for the bootstrap.

For $p=10$, the actual coverage probability of the confidence intervals for both our algorithm and bootstrap is approximately equal to the nominal coverage probability.
The confidence intervals constructed by our algorithm are conservative, but the actual coverage probability approaches the nominal coverage probability as $r$ increases. The latter is expected, since for larger $r$ the approximation $W \approx \mathbf{1}_N$ made in Section \ref{sec:infsigmaunknown} should be more accurate.
The conservativeness of the confidence intervals constructed by our algorithm seems to increase with $p$. 
The bootstrap confidence intervals are generally less conservative than those from our algorithm, but may have less than the nominal coverage probability as shown in Figure \ref{fig:coveragep10n5000}.

The type 1 error of our algorithm is generally below that of bootstrap, especially for smaller $r$.
While type 1 error of our algorithm is generally smaller than the nominal $\alpha=0.05$, the type 1 error of the bootstrap may be much more, in particular for $p=10$. This is consistent with the fact that the bootstrap coverage probability may be less than the nominal value.
It also appears that for fixed $N$ and $r$, as $p$ increases the type 1 errors of our algorithm and of bootstrap decreases.

Analogously, the type 2 error of our algorithm is usually above that of bootstrap.
However, for both approaches, the type 2 error is decreasing with $r$, below $0.2$ for $r\geq 300$, and close to $0$ at $r=N$.
Therefore, as long as we don't use too small samples, we can obtain an estimate of $\beta$ in less time without sacrificing much power in our significance tests.
As $p$ increases the type 2 errors of both approaches increase.

\begin{figure}[ht]
	\centering
	\subfigure[Coverage Probabilities]{%
		\includegraphics[width=0.45\textwidth]{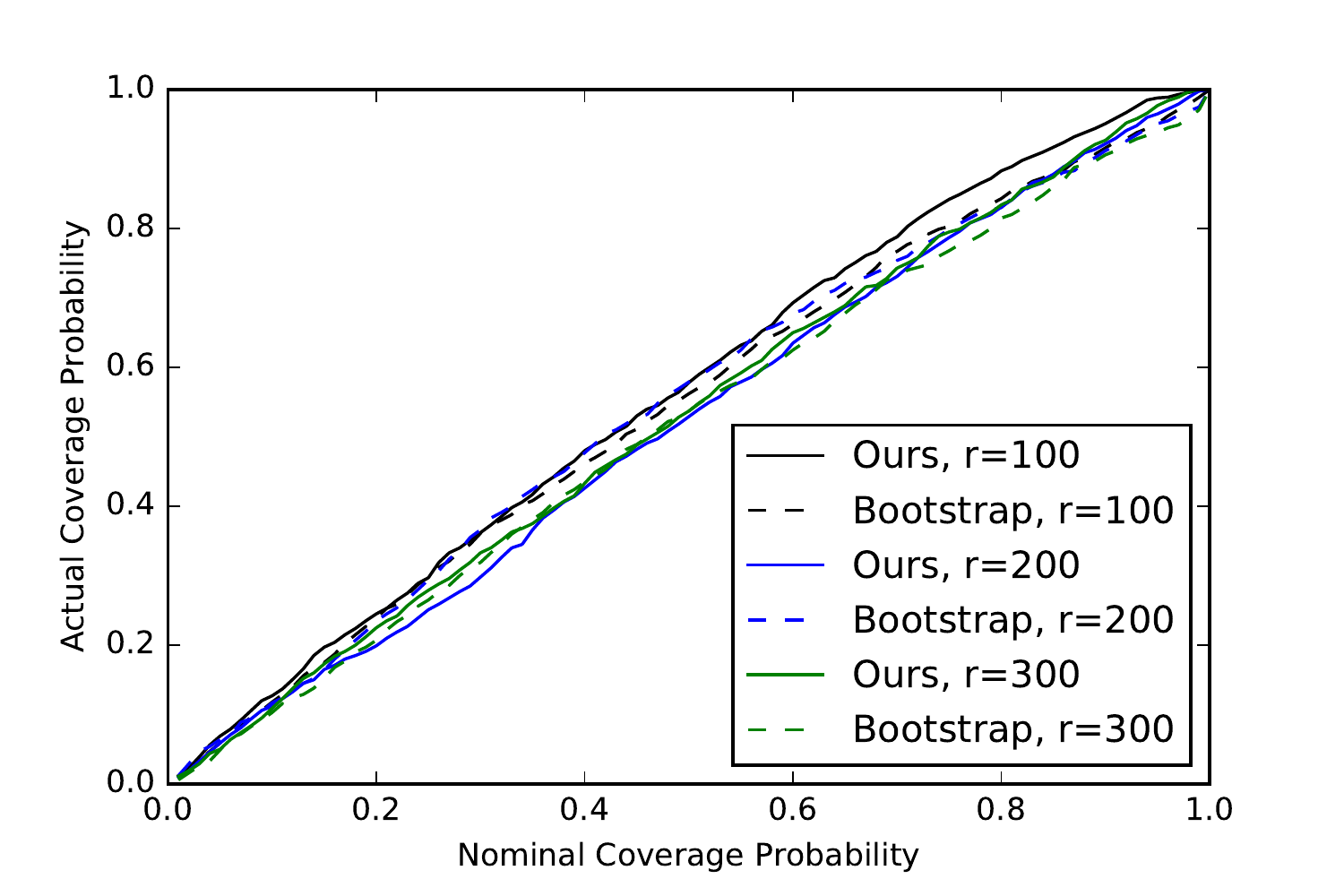}\label{fig:coveragep10n1000}
	}
	\subfigure[Computation Time]{%
		\includegraphics[width=0.45\textwidth]{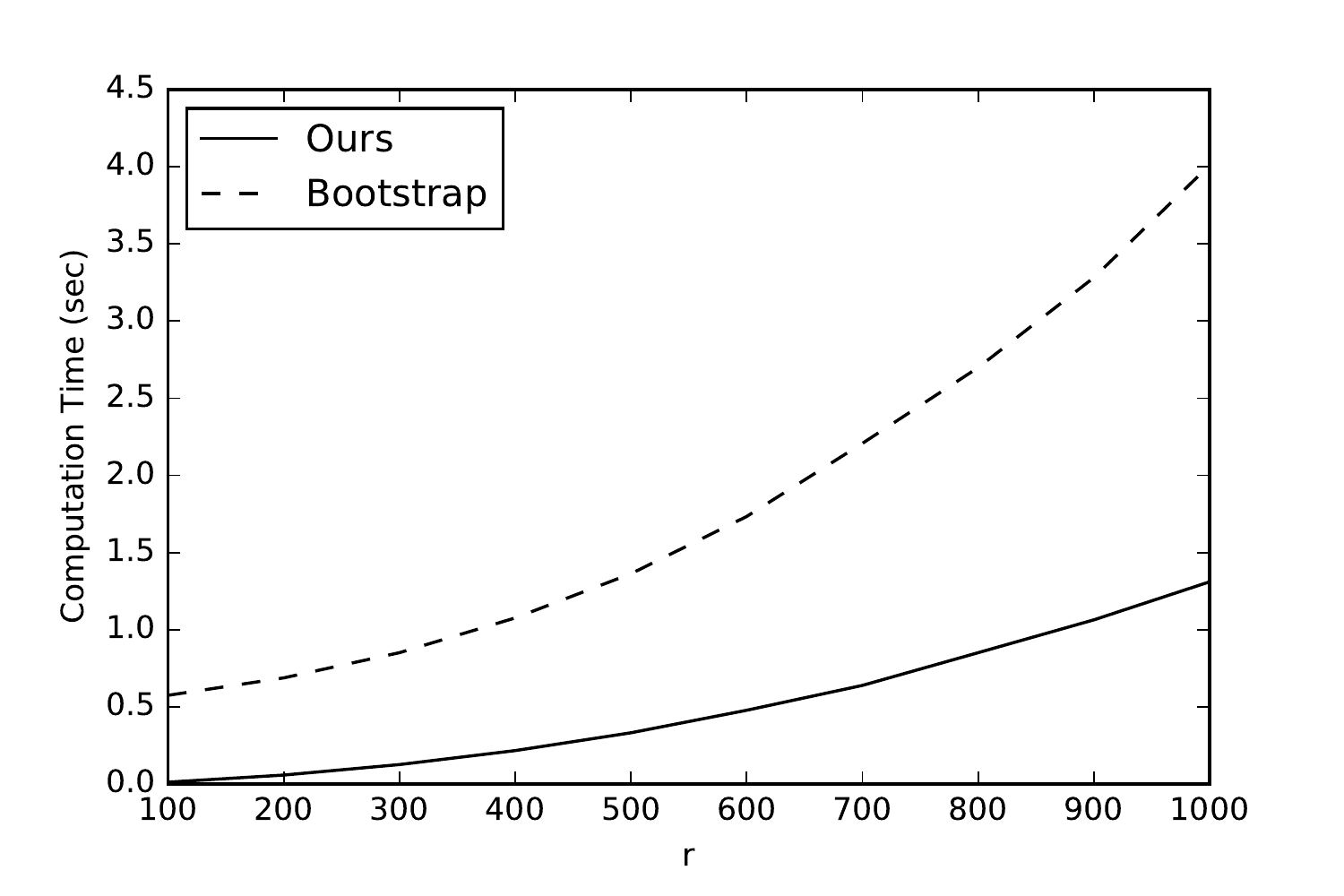}\label{fig:cpup10n1000}
	}
	\\
	\subfigure[Type 1 Error]{%
		\includegraphics[width=0.45\textwidth]{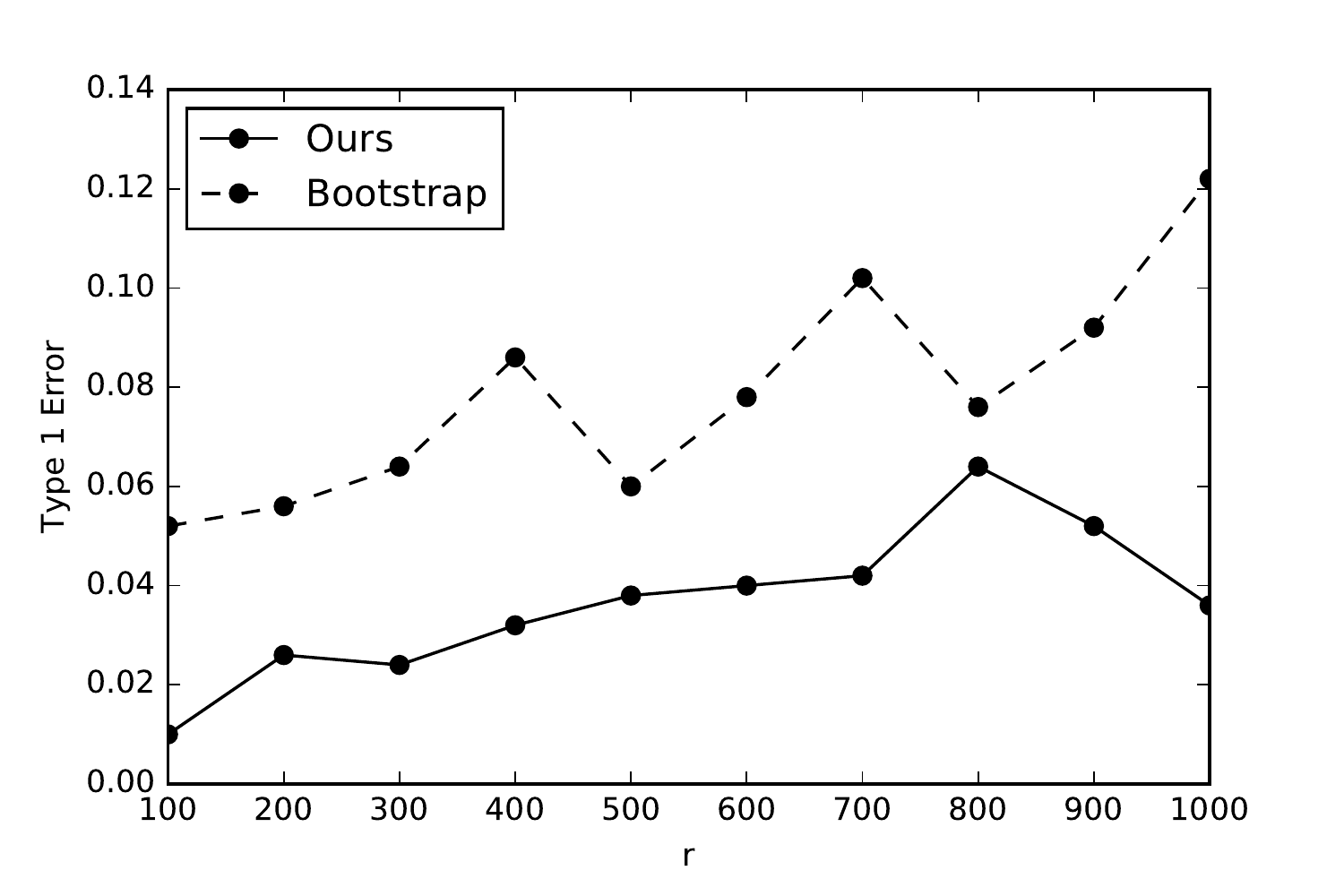}\label{fig:type1p10n1000}
	}
	\subfigure[Type 2 Error]{%
		\includegraphics[width=0.45\textwidth]{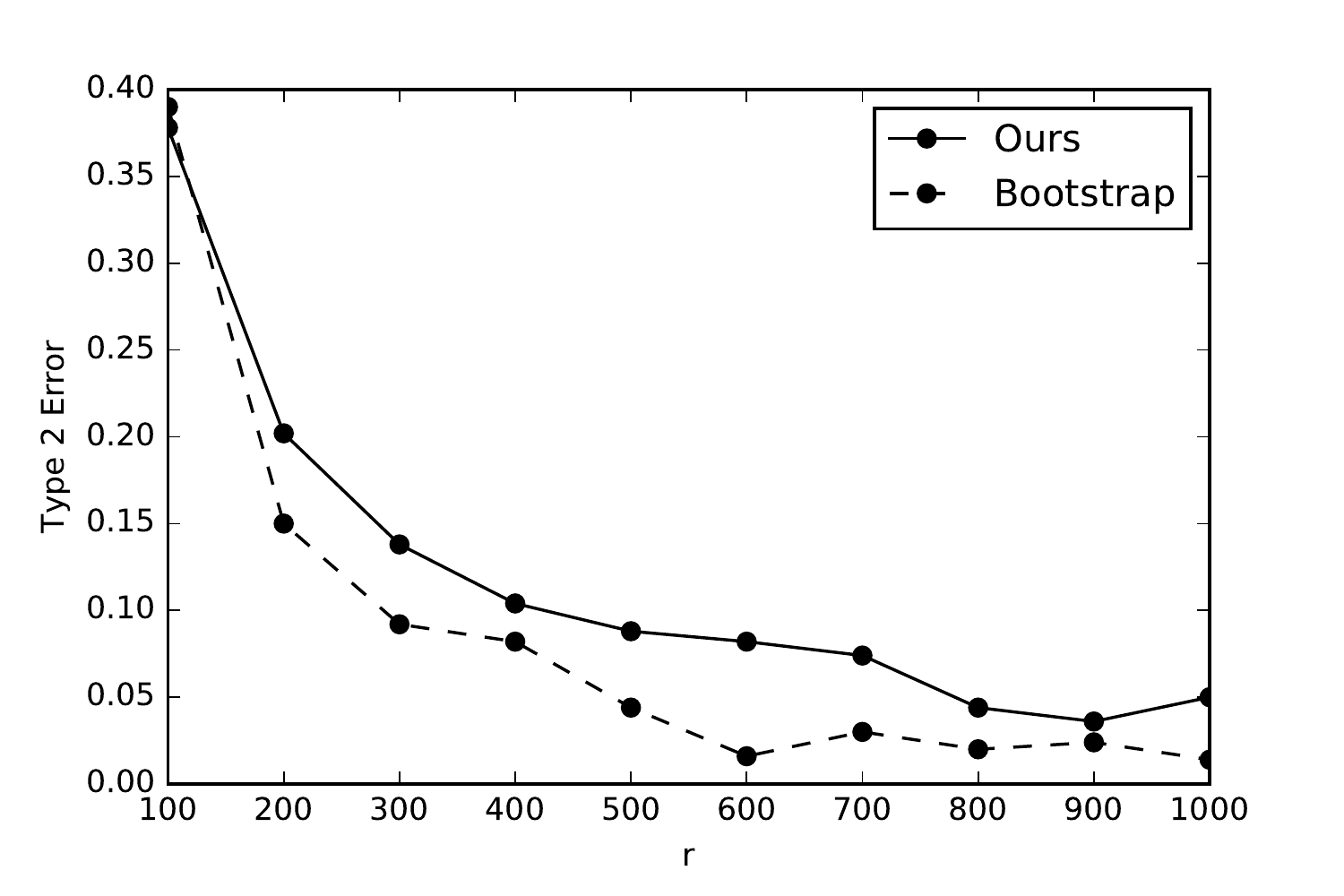}\label{fig:type2p10n1000}
	}
	\caption{$p=10,N=1000$}
	\label{fig:p10N1000}
\end{figure}

\begin{figure}[ht]
	\centering
	\subfigure[Coverage Probabilities]{%
		\includegraphics[width=0.45\textwidth]{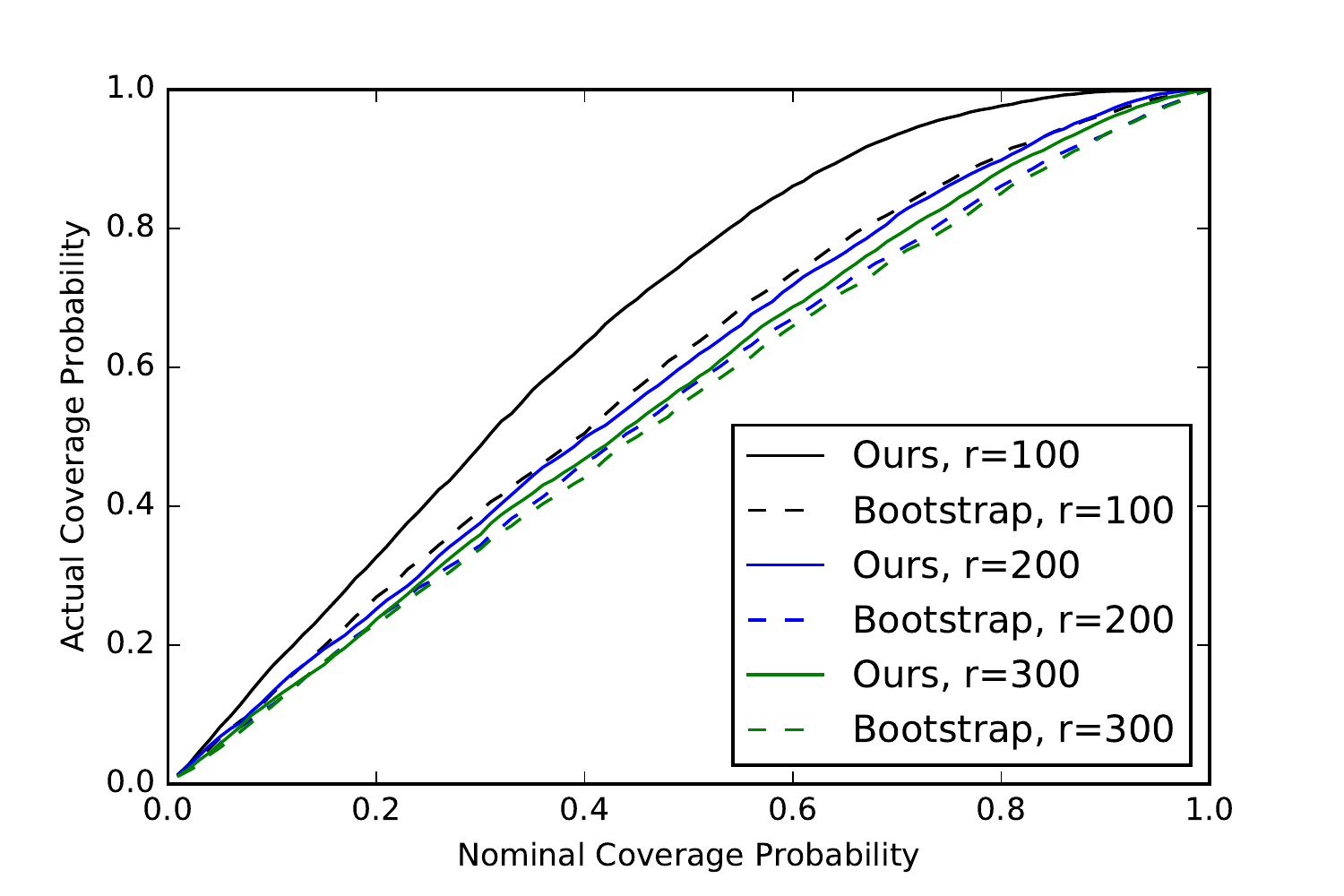}\label{fig:coveragep50n1000}
	}
	\subfigure[Computation Time]{%
		\includegraphics[width=0.45\textwidth]{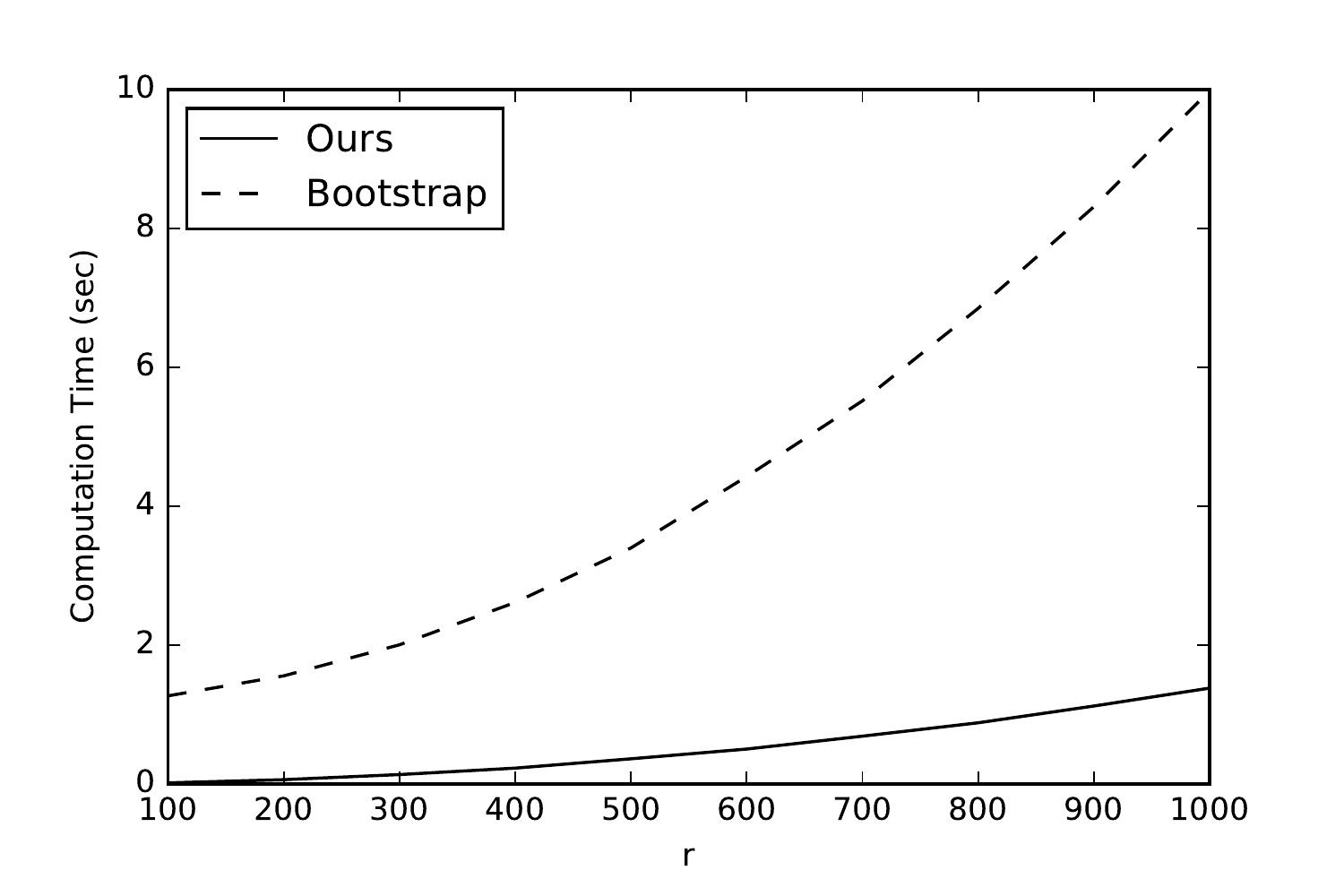}\label{fig:cpup50n1000}
	}
	\\
	\subfigure[Type 1 Error]{%
		\includegraphics[width=0.45\textwidth]{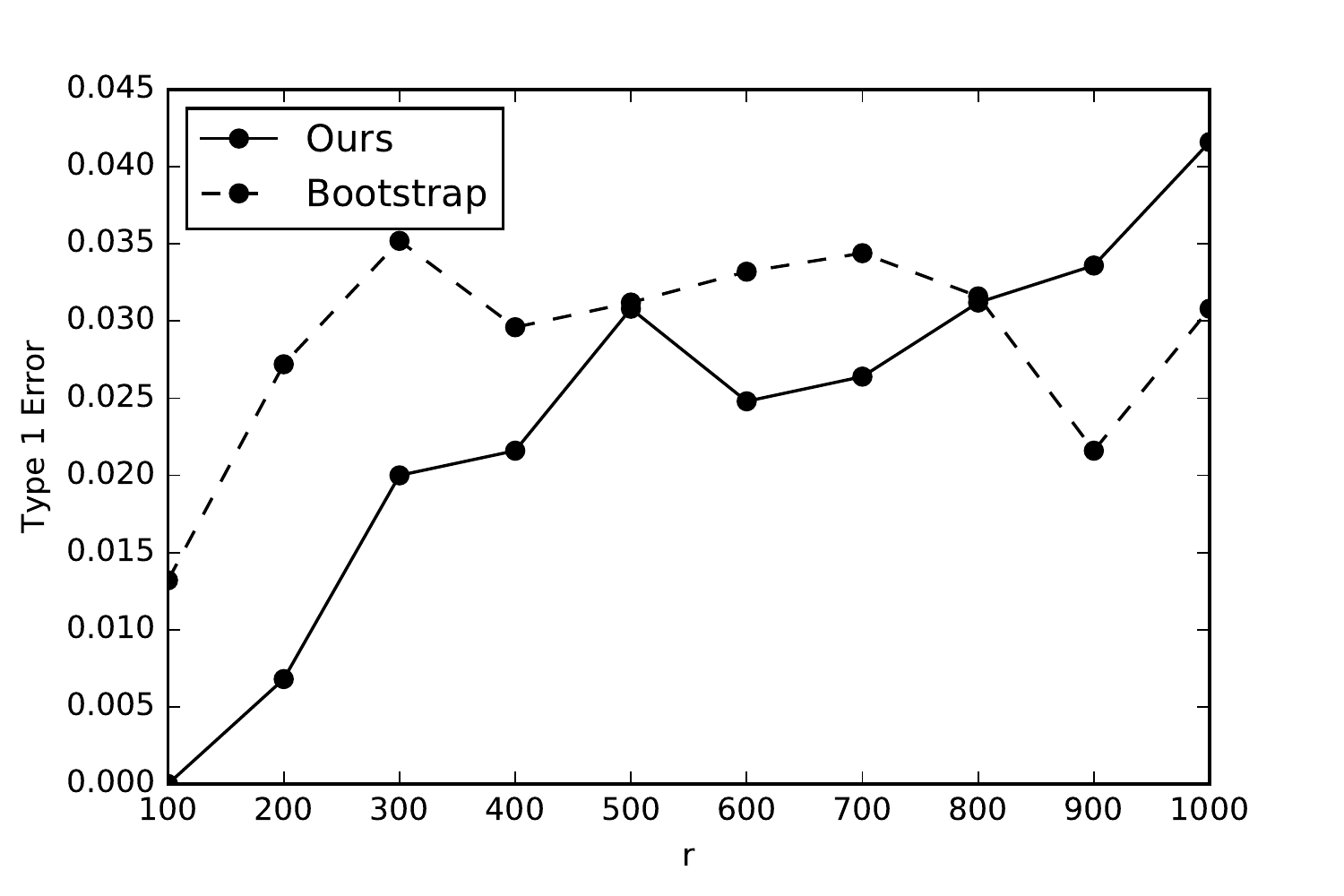}\label{fig:type1p50n1000}
	}
	\subfigure[Type 2 Error]{%
		\includegraphics[width=0.45\textwidth]{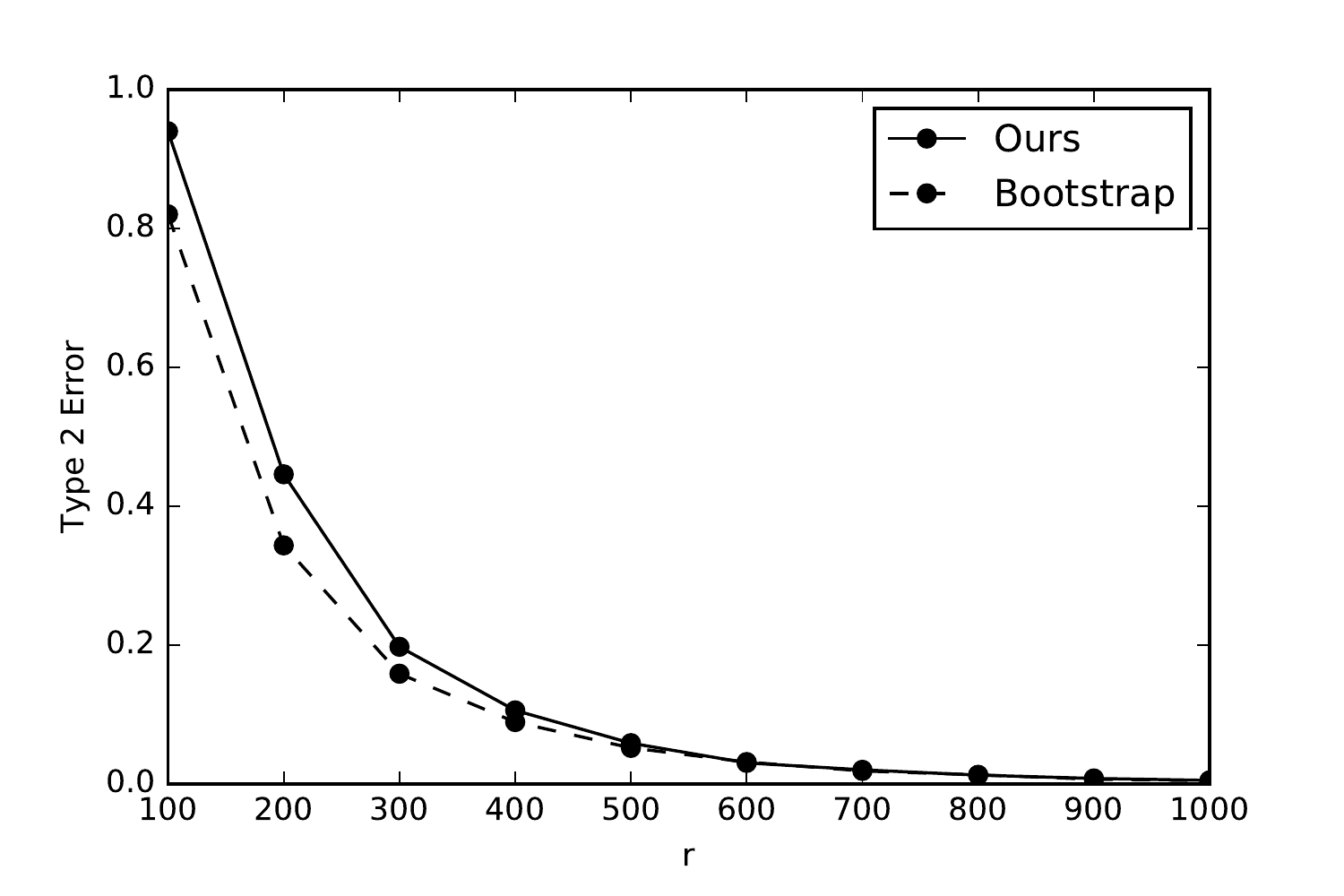}\label{fig:type2p50n1000}
	}
	\caption{$p=50,N=1000$}
	\label{fig:p50N1000}
\end{figure}

\begin{figure}[ht]
	\centering
	\subfigure[Coverage Probabilities]{%
		\includegraphics[width=0.45\textwidth]{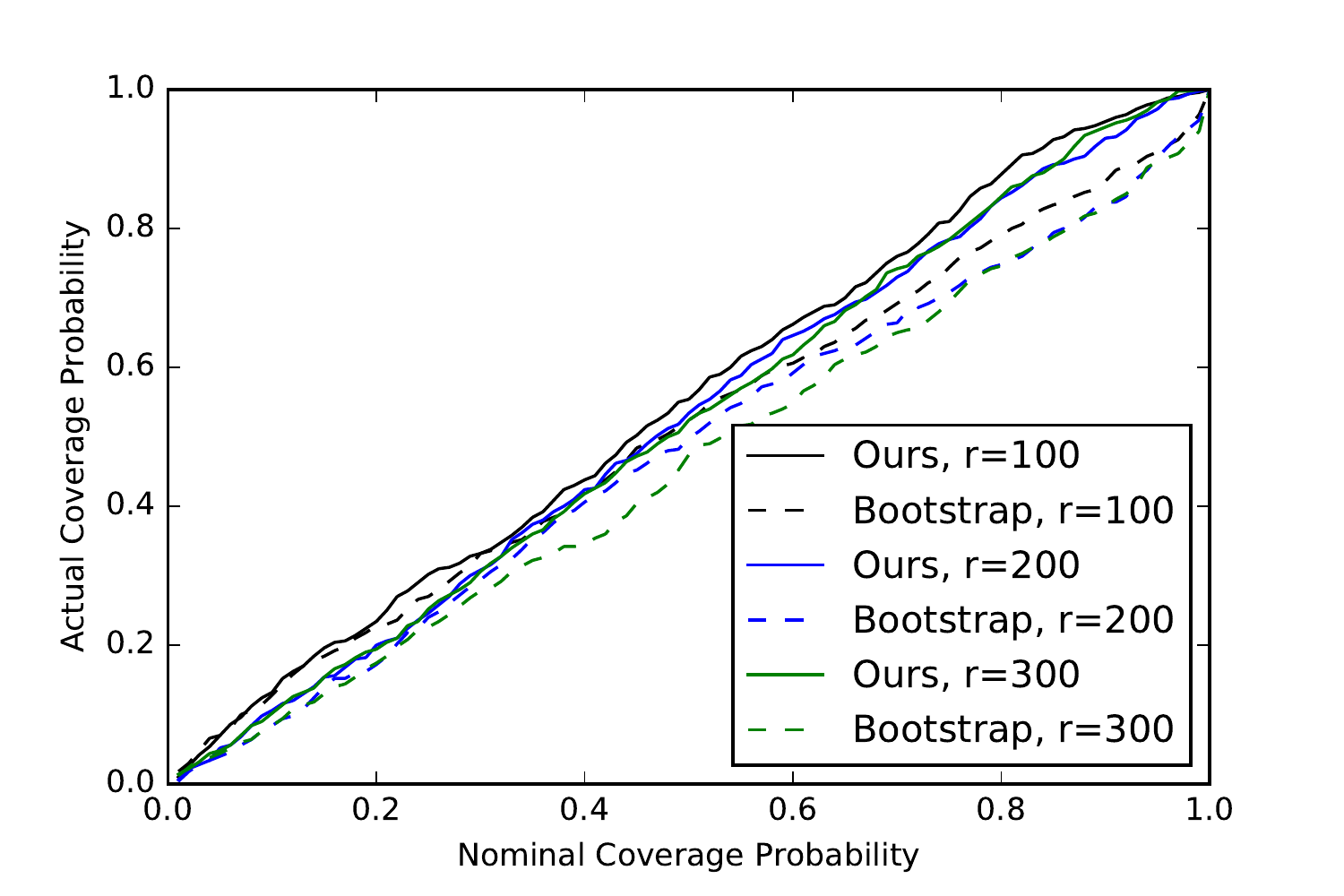}\label{fig:coveragep10n5000}
	}
	\subfigure[Computation Time]{%
		\includegraphics[width=0.45\textwidth]{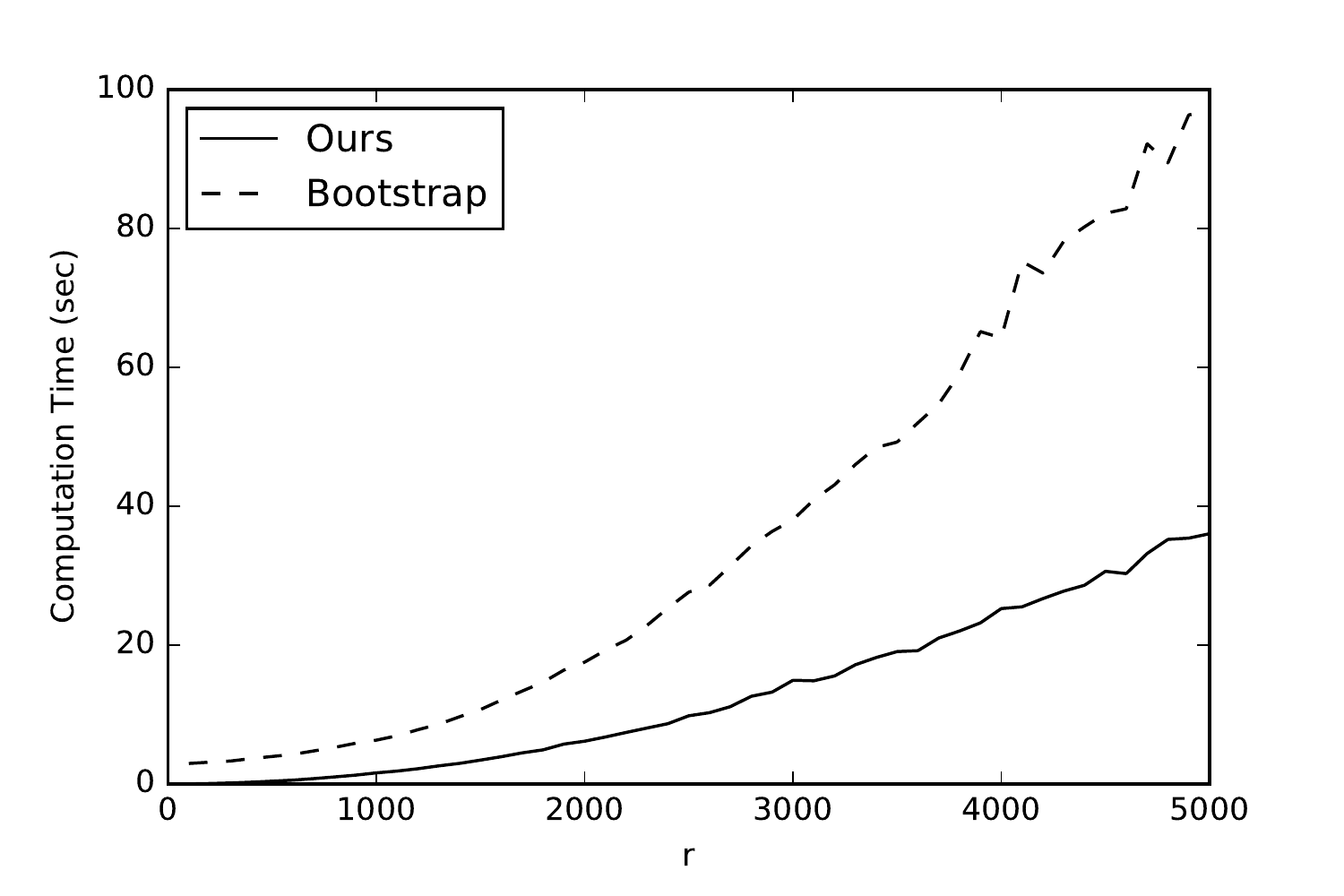}\label{fig:cpup10n5000}
	}
	\\
	\subfigure[Type 1 Error]{%
		\includegraphics[width=0.45\textwidth]{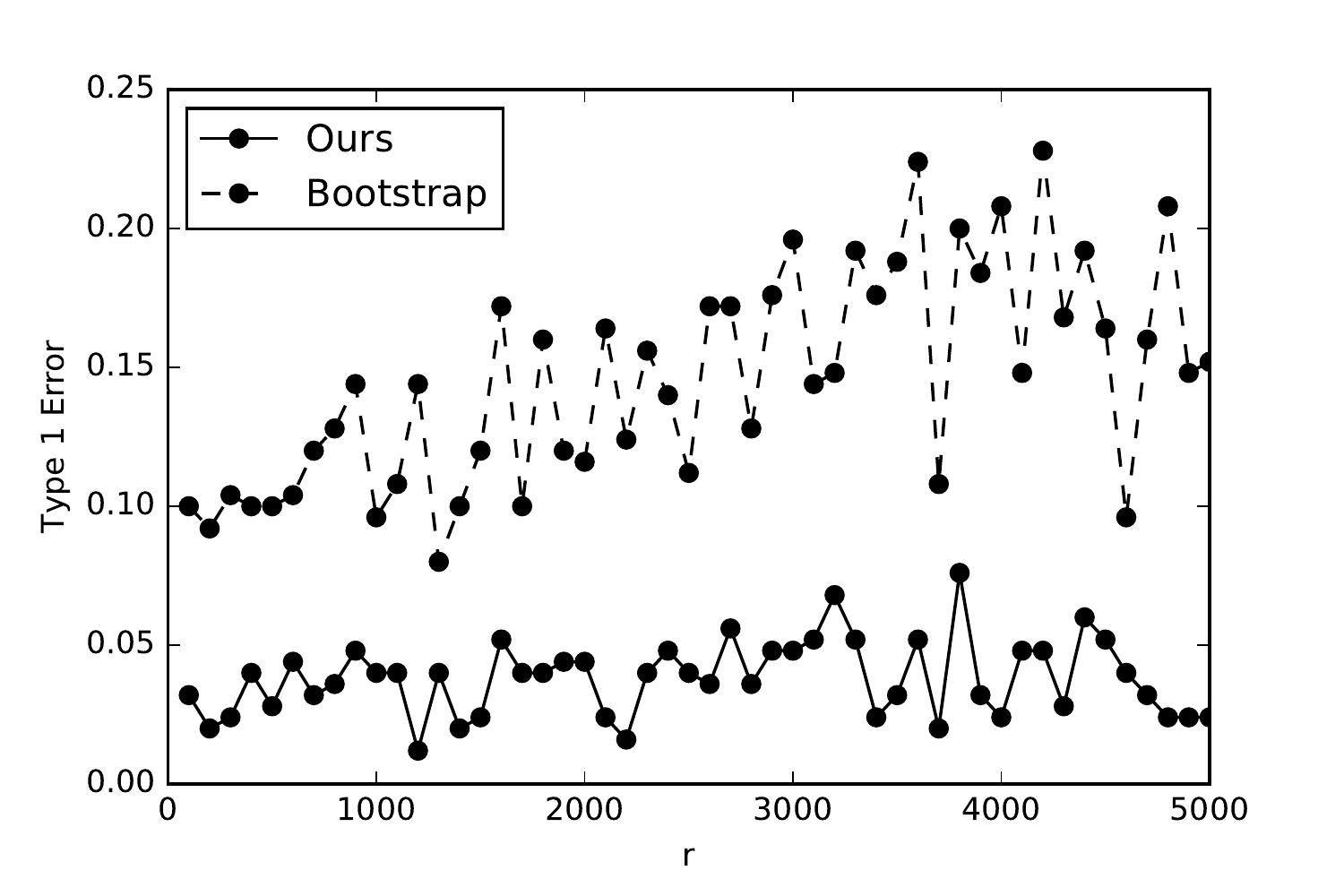}\label{fig:type1p10n5000}
	}
	\subfigure[Type 2 Error]{%
		\includegraphics[width=0.45\textwidth]{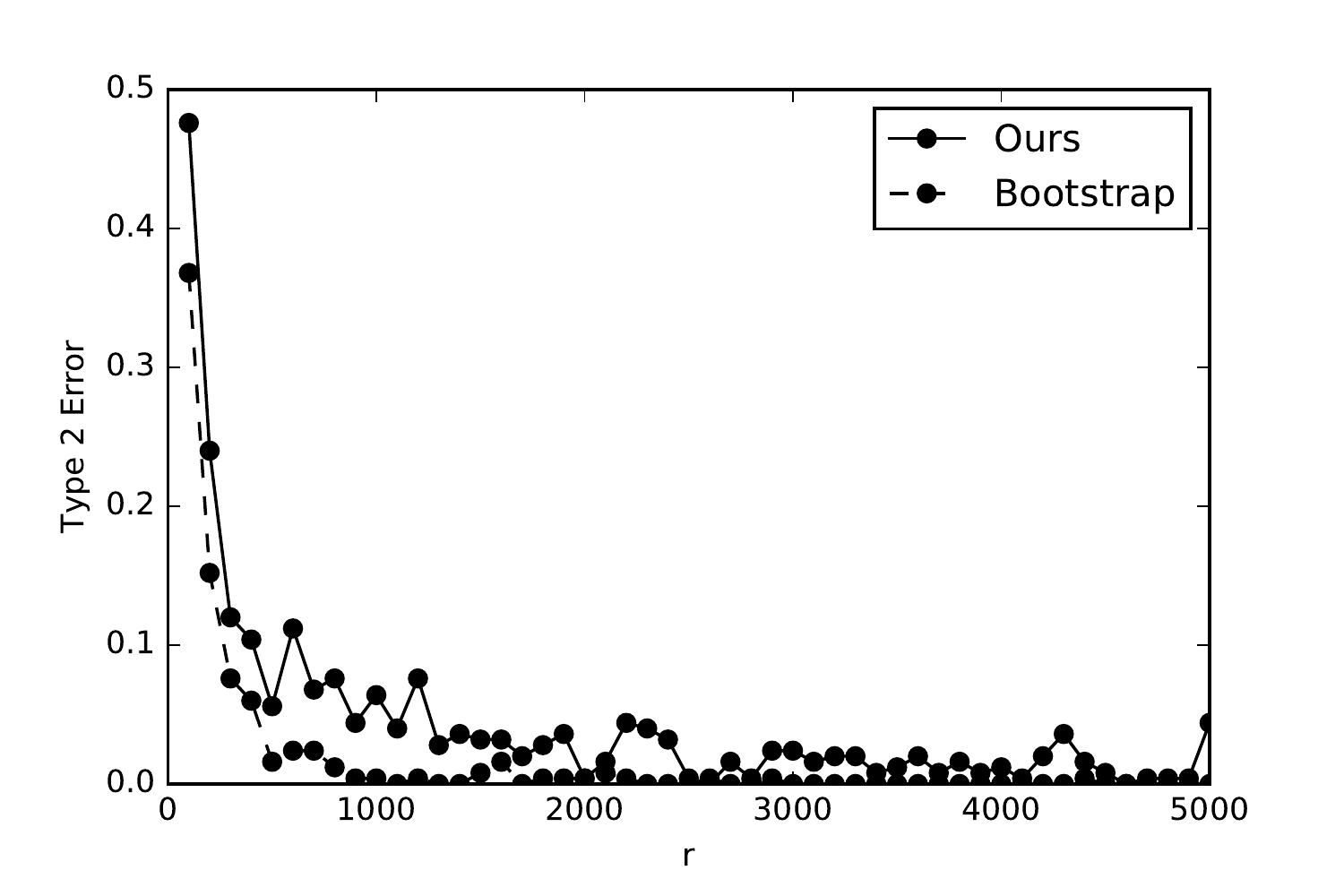}\label{fig:type2p10n5000}
	}
	\caption{$p=10,N=5000$}
	\label{fig:p10N5000}
\end{figure}

\begin{figure}[ht]
	\centering
	\subfigure[Coverage Probabilities]{%
		\includegraphics[width=0.45\textwidth]{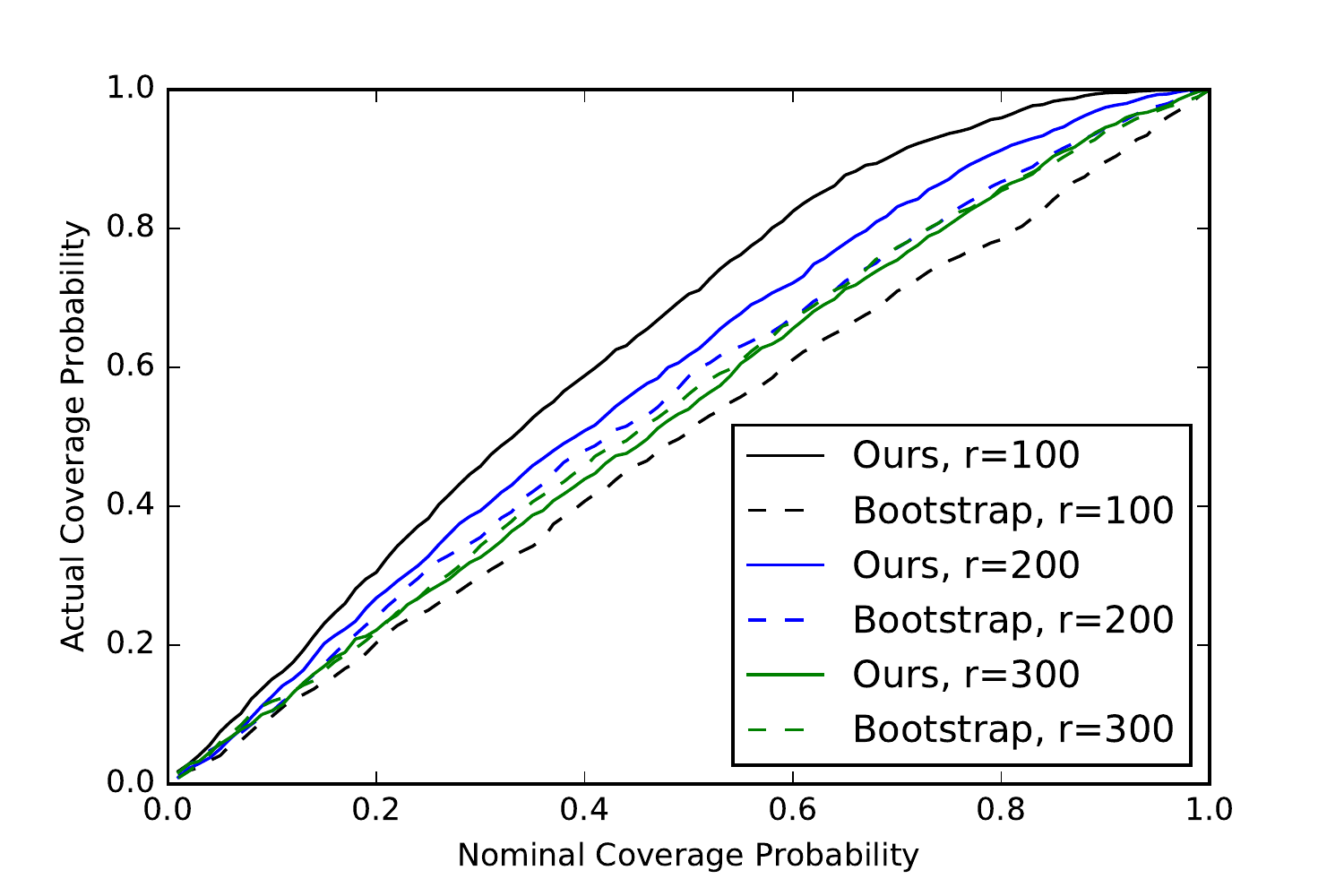}\label{fig:coveragep50n5000}
	}
	\subfigure[Computation Time]{%
		\includegraphics[width=0.45\textwidth]{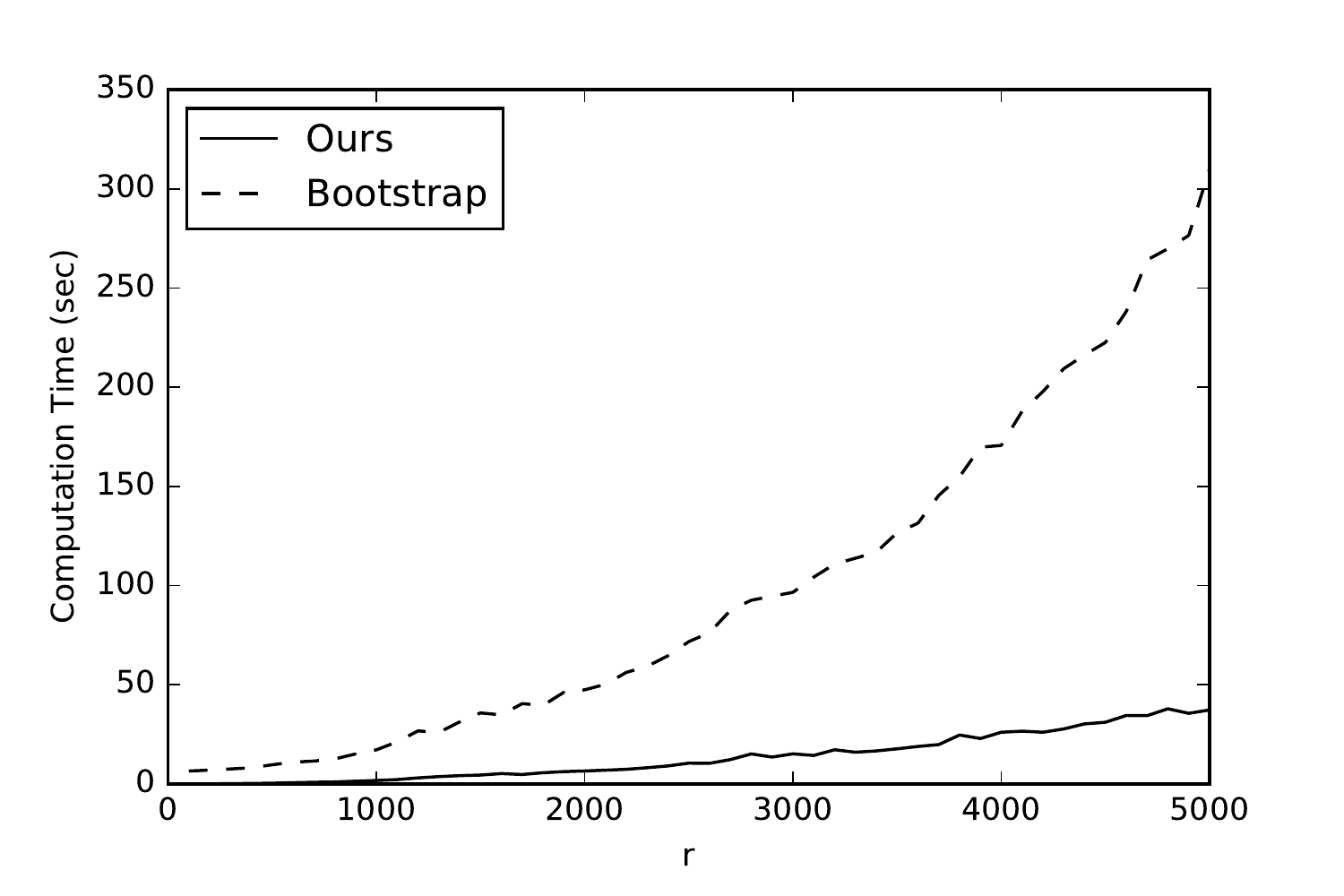}\label{fig:cpup50n5000}
	}
	\\
	\subfigure[Type 1 Error]{%
		\includegraphics[width=0.45\textwidth]{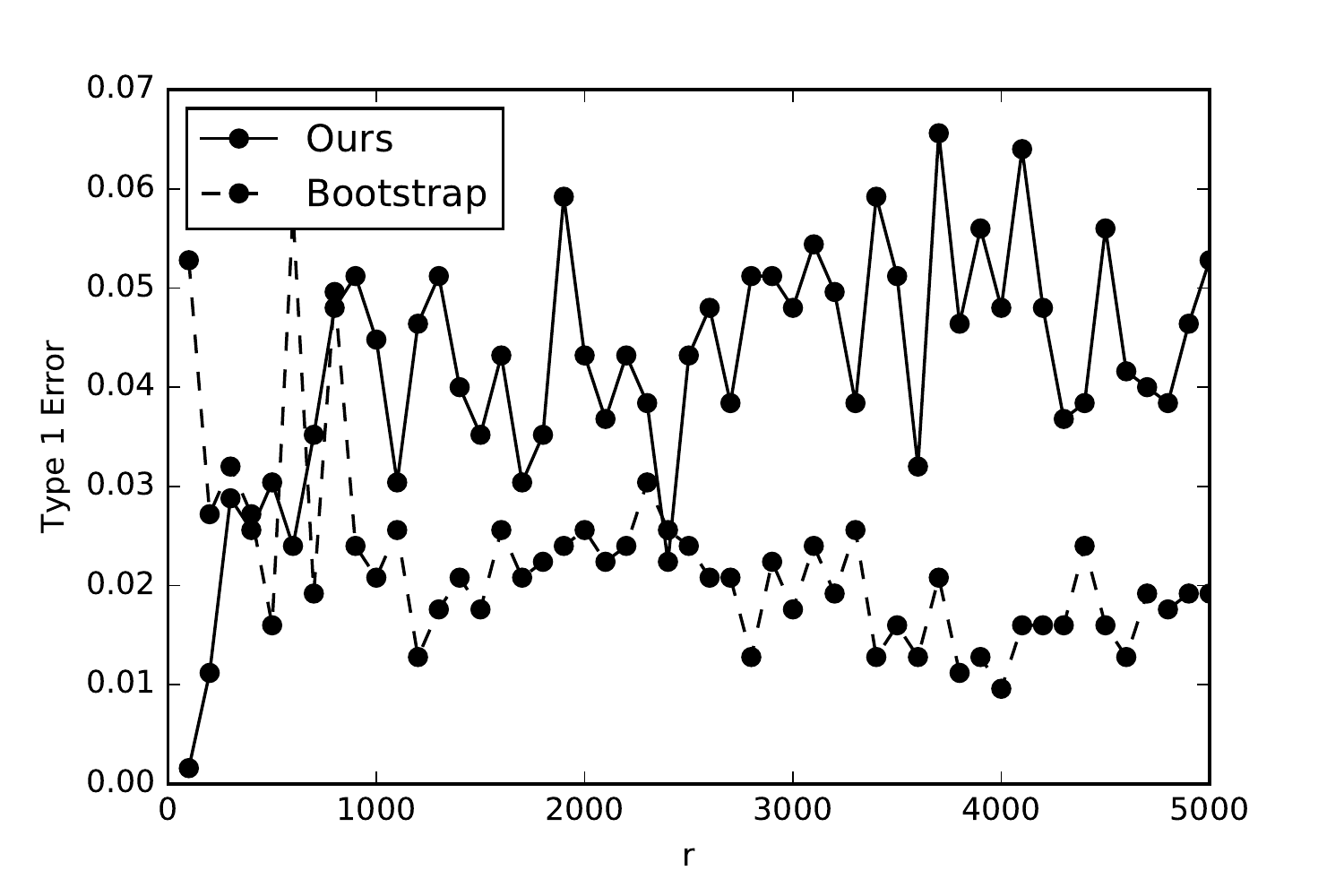}\label{fig:type1p50n5000}
	}
	\subfigure[Type 2 Error]{%
		\includegraphics[width=0.45\textwidth]{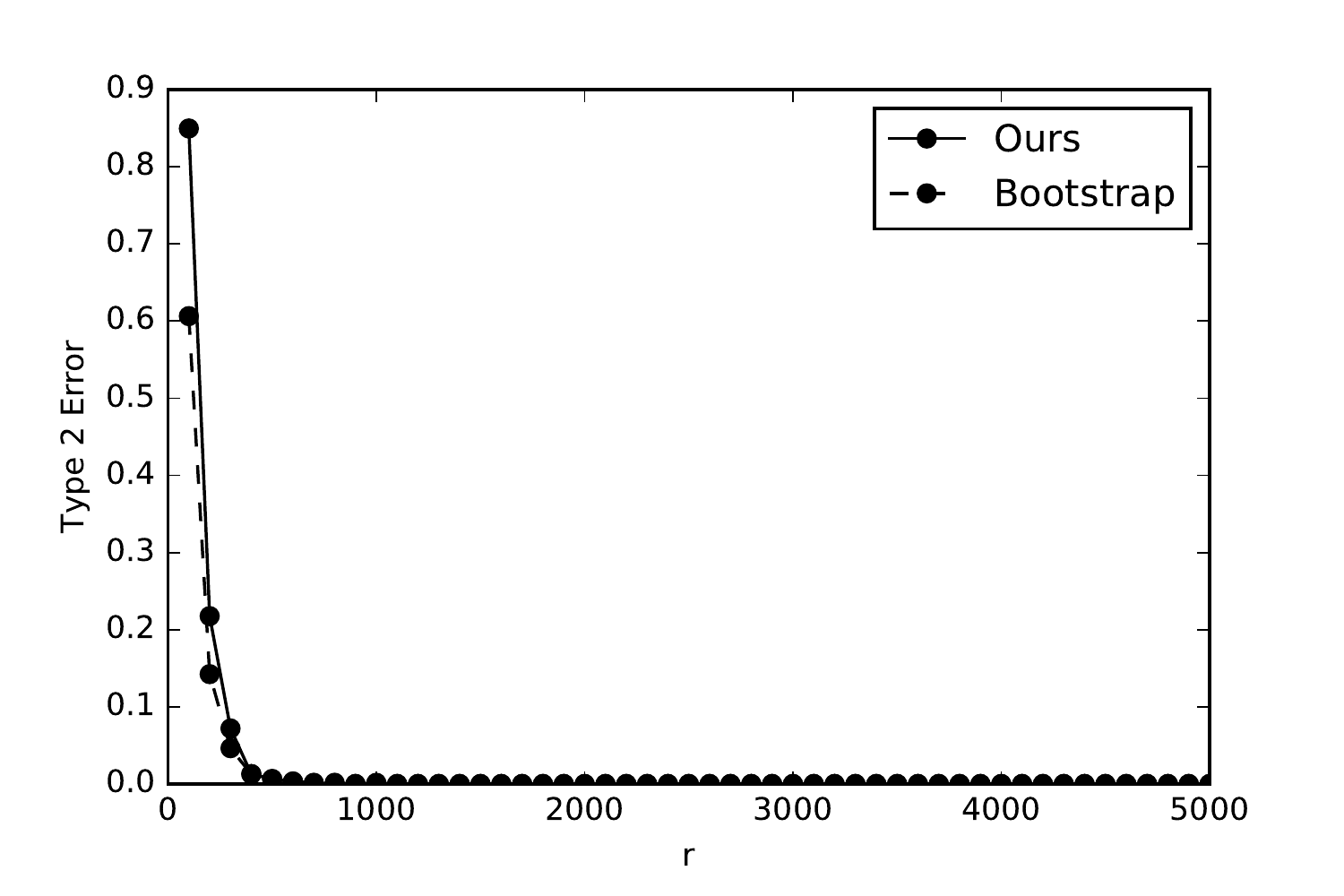}\label{fig:type2p50n5000}
	}
	\caption{$p=50,N=5000$}
	\label{fig:p50N5000}
\end{figure}

\section{Conclusion and Future Work}\label{sec:conclusion}

Learning from and mining massive data sets post great challenges given our limited storage and computational resources. Many data reduction approaches have been devised to overcome such challenges. Algorithmic leveraging is such an approach.
In this paper, for linear regression coefficients estimated using algorithmic leveraging, we described how to efficiently construct finite sample confidence intervals and significance tests. Simulations show that our proposed confidence intervals have the desired coverage probability and our proposed significance tests control the type 1 error rate and have low type 2 error rates.
The simulations also show that bootstrap confidence intervals may have smaller than the desired coverage probability. 

There are several avenues for future work investigating the statistical properties of algorithmic leveraging applied to data analyses beyond simple linear regression. For instance, we believe that determining how sampling affects feature selection in Lasso regression \cite{T96} could have important practical implications. Finally, we may consider uncertainty quantification for estimates from other data reduction methods, such as sketching algorithms, which are another popular method to deal with massive data sets.

\section*{Acknowledgements}

The author was supported by US NSF under grants DGE-114747, DMS-1407397, and DMS-1521145. 

\bibliographystyle{plain}
\bibliography{references}

\section{Appendix}

\subsection{Proof of Lemma \ref{lemma:wasconv}}\label{sec:proofs:wasconv}

We apply the strong law of large numbers \cite{B12}. Denote the $i$th row of a matrix $M$ by $(M)_{i\cdot}$ and the $j$th column by $(M)_{\cdot j}$. Then, 
\begin{align*}
w_j &= (S_X)_{j\cdot}D^2(S_X^T)_{\cdot j} = \sum_{i=1}^r\dfrac{(S_X)_{j,i}^2}{r\pi_j}=\dfrac{1}{r}\sum_{i=1}^r \dfrac{(S_X)_{j,i}}{\pi_j} \\
\mathbf{w} &= \dfrac{1}{r}\sum_{i=1}^r (\dfrac{(S_X)_{1,i}}{\pi_1},\ldots,\dfrac{(S_X)_{N,i}}{\pi_N})^\tran
\end{align*}
Thus, $\mathbf{w}$ is the average of $r$ independent and identically distributed random variables. Each has mean $\mathbf{1}_N$, since $(S_X)_{ji}$ is $1$ with probability $\pi_j$. The result follows.

\subsection{Proof of Theorem \ref{thm:betaclt}}\label{sec:proofs:betaclt}

We first need the following lemma from Ma et al. \cite{MMY15}. 
\begin{lemma}\label{lemma:wmoments}
\begin{align}
\e(\mathbf{w}) &= \mathbf{1}_N \\
\var(\mathbf{w}) &= \dfrac{1}{r}(diag(\dfrac{1}{\mathbf{\pi}})-\mathbf{1}_N\mathbf{1}_N^\tran)
\end{align}
where $\mathbf{1}_N$ is the $N$-dimensional vector of ones and ${\mathbf{\pi}=(\pi_1,...,\pi_N)^\tran}$.
\end{lemma}

Since $\var(\mathbf{w}) = O(1/r)$, it follows that $\mathbf{w}$ converges in probability to $\mathbf{1}_N$ as ${r \to \infty}$ with $N$ fixed. 

We rewrite $\betahw$. Let $\epsilon$ be the $N$-dimensional vector of $\epsilon_i$'s. 
\begin{align*}
\betahw &= (X^\tran WX)^{-1}X^\tran WY \\
&= (X^\tran WX)^{-1}X^\tran W(X\beta+\epsilon) \\
&= \beta+(X^\tran WX/N)^{-1}X^\tran W\epsilon/N
\end{align*}
By Slutsky's Theorem \cite{B12} and Lemma 1, 
as $r\to\infty$ for fixed $N$, $X^\tran WX/N$ converges in probability to $X^\tran X/N$. We have assumed that as $N\to\infty$, $X^\tran X/N$ converges to a finite, positive definite matrix $M$.

Therefore, as $r, N\to\infty$, $(X^\tran WX/N)^{-1}$ converges in probability to $M^{-1}$.
In order to show that $\betahw$ is asymptotically normal, we just have to show that $X^\tran W\epsilon/N$ is asymptotically normal. 

With $N$ fixed, as $r\to\infty$, $X^\tran W\epsilon/N$ converges in probability to $X^\tran\epsilon/N$.
\begin{align*}
\dfrac{X^\tran\epsilon}{N} &= \dfrac{1}{N}\sum_{i=1}^N x_i\epsilon_i
\end{align*}
where $x_i\epsilon_i$ has mean zero and variance $\sigma^2 x_i x_i^\tran$. We have assumed that as $N\to\infty$, $\sigma^2\sum_{i=1}^N x_i x_i^\tran/N$ converges to $\sigma^2 M$. Assuming that for each $i=1,\ldots,N$,
\begin{align*}
\lim_{N\to\infty}(\sum_{i=1}^N x_i x_i^\tran)^{-1}x_i x_i^\tran = 0,
\end{align*}
by the multivariate Lindeberg-Feller central limit theorem \cite{greene}, $\dfrac{1}{N}\sum_{i=1}^N x_i\epsilon_i$ is asymptotically normal.

Hence, as $r,N\to\infty$, $X^\tran W\epsilon/N$ is normally distributed. Then, $\betahw$ is asymptotically Gaussian. Its mean and variance are computed from approximate formulas given in Ma et al. \cite{MMY15}, using the fact that as $r\to\infty$ and $N$ is fixed $\mathbf{w}$ converges in probability to $\mathbf{1}_N$. 

\subsection{Proof of Theorem \ref{thm:CIsigmaknown}}\label{sec:proofs:CIsigmaknown}

From \eqref{eq:betahwexp} and \eqref{eq:betahwvar}, we have that
\begin{align*}
1-\alpha &=\mathbb{P}(|(\betahw)_j-\beta_j|\leq z_{1-\alpha/2}\sigma\sqrt{((X^\tran WX)^{-1}X^\tran W^2X(X^\tran WX)^{-1})_{jj}}\mid W)
\end{align*}
Taking expectation of the above equation with respect to $W$, we have
\begin{align*}
1-\alpha &= \mathbb{P}(|(\betahw)_j-\beta_j|\leq z_{1-\alpha/2}\sigma\sqrt{((X^\tran WX)^{-1}X^\tran W^2X(X^\tran WX)^{-1})_{jj}})
\end{align*}
from which the result follows.

\subsection{Proof of Theorem \ref{thm:CIsigmaunknown}}\label{sec:proofs:CIsigmaunknown}

From Lemma \ref{lemma:wasconv}, as $r\to\infty$, $\mathbf{w} \to \mathbf{1}_N$ almost surely. Then, with probability $1$, $H_W \to H$. Therefore, as $r\to\infty$, 

\begin{align*}
(\betahw-\beta,(I-H_W)\epsilon) \overset{\mathcal{D}}{\to} (\betahols-\beta, (I-H)\epsilon)
\end{align*}

Considering only the $j$th predictor and scaling appropriately, we have

\begin{align*}
\left(\dfrac{(\betahw)_j-\beta_j}{\sqrt{((X^\tran WX)^{-1}X^\tran W^2X(X^\tran WX)^{-1})_{jj}}},\dfrac{(I-H_W)\epsilon}{\sqrt{N-p}}\right) \overset{\mathcal{D}}{\to} \left(\dfrac{(\betahols)_j-\beta_j}{\sqrt{((X^\tran X)^{-1})_{jj}}}, \dfrac{(I-H)\epsilon}{\sqrt{N-p}}\right)
\end{align*}

Applying the continuous mapping theorem, we have 

\begin{align*}
& \dfrac{(\betahw)_j-\beta_j}{\sqrt{\|(I-H_W)\epsilon\|^2/(N-p)}\sqrt{((X^\tran WX)^{-1}X^\tran W^2X(X^\tran WX)^{-1})_{jj}}} \\
& \overset{\mathcal{D}}{\to} \dfrac{(\betahols)_j-\beta_j}{\sqrt{\|(I-H)\epsilon\|^2/(N-p)}\sqrt{((X^\tran X)^{-1})_{jj}}} \\
& \sim t_{N-p}
\end{align*}
where the last line follows from \cite{W05} assuming that $X$ is full rank.

Therefore, as $r\to\infty$, $(\betahw)_j \pm t_{N-p,1-\alpha/2}\hat\sigma \sqrt{((X^\tran WX)^{-1}X^\tran W^2X(X^\tran WX)^{-1})_{jj}}$ has coverage probability $1-\alpha$. 

\subsection{ROC Curves}\label{sec:experiments:roc}

Here are the ROC curves for the simulations in the main paper.

\begin{figure}[ht]
	\centering
	\subfigure[$p=10$, $N=1000$]{%
		\includegraphics[width=0.45\textwidth]{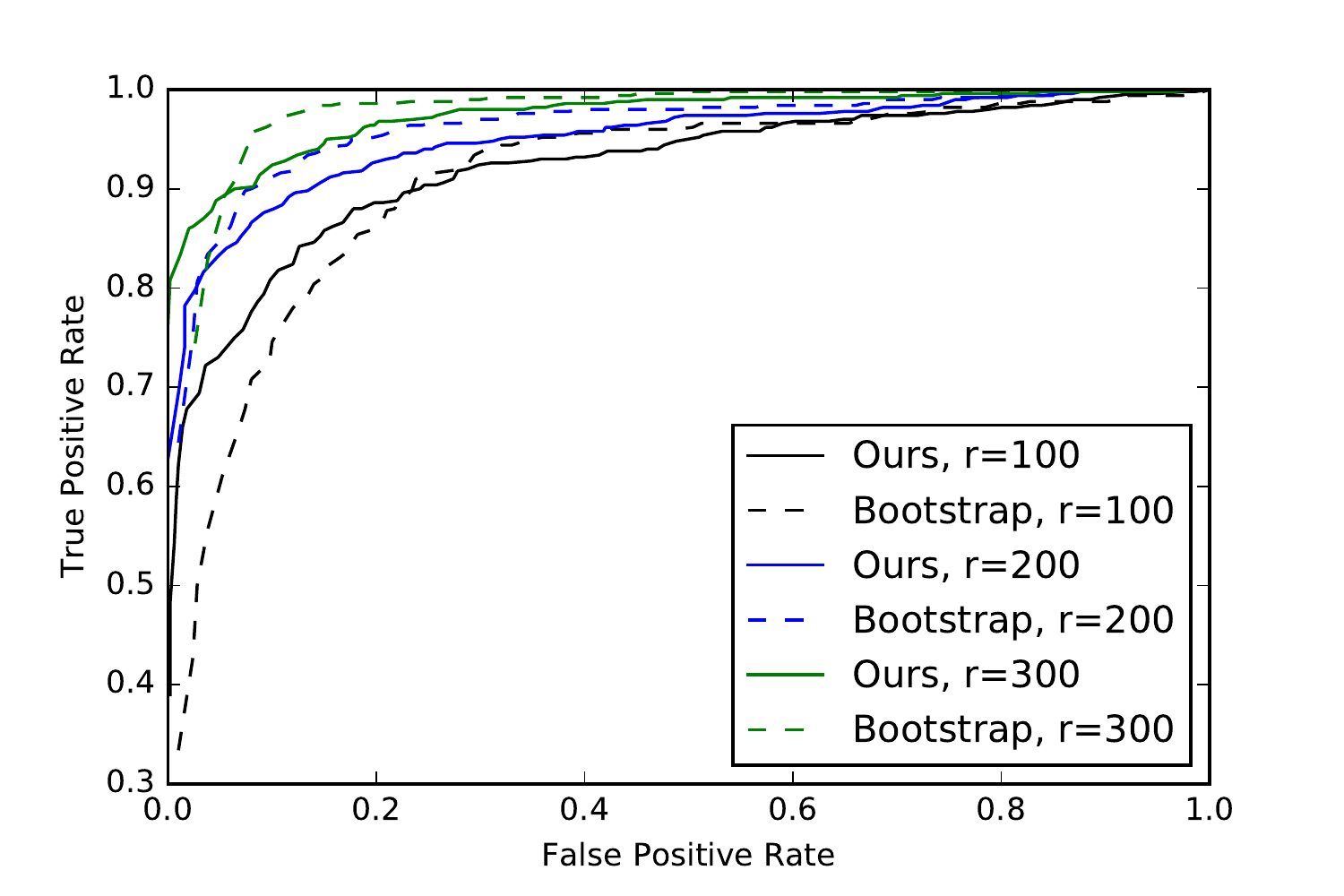}\label{fig:rocp10n1000}
	}
	\subfigure[$p=50$, $N=1000$]{%
		\includegraphics[width=0.45\textwidth]{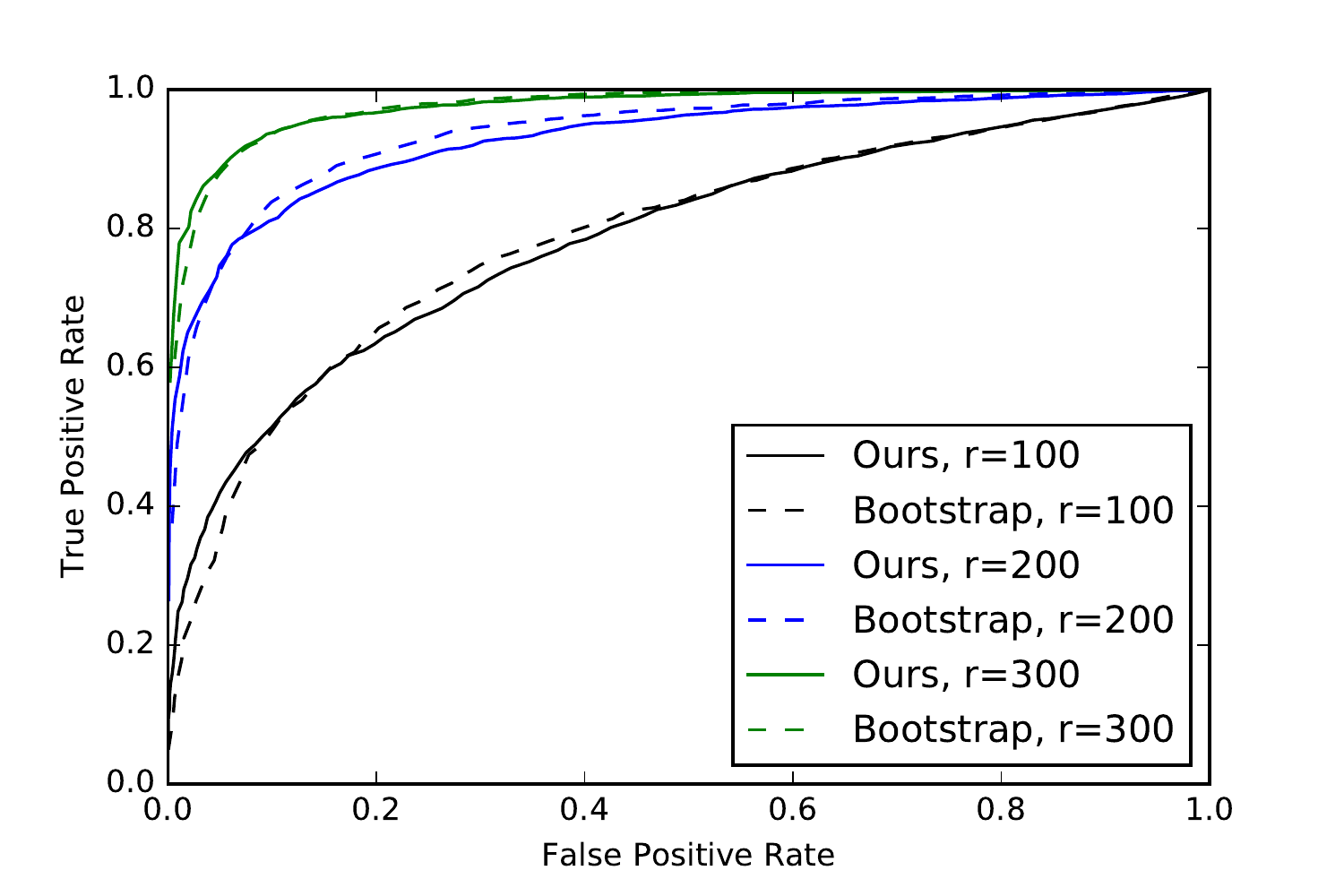}\label{fig:rocp50n1000}
	}
	\\
	\subfigure[$p=10$, $N=5000$]{%
		\includegraphics[width=0.45\textwidth]{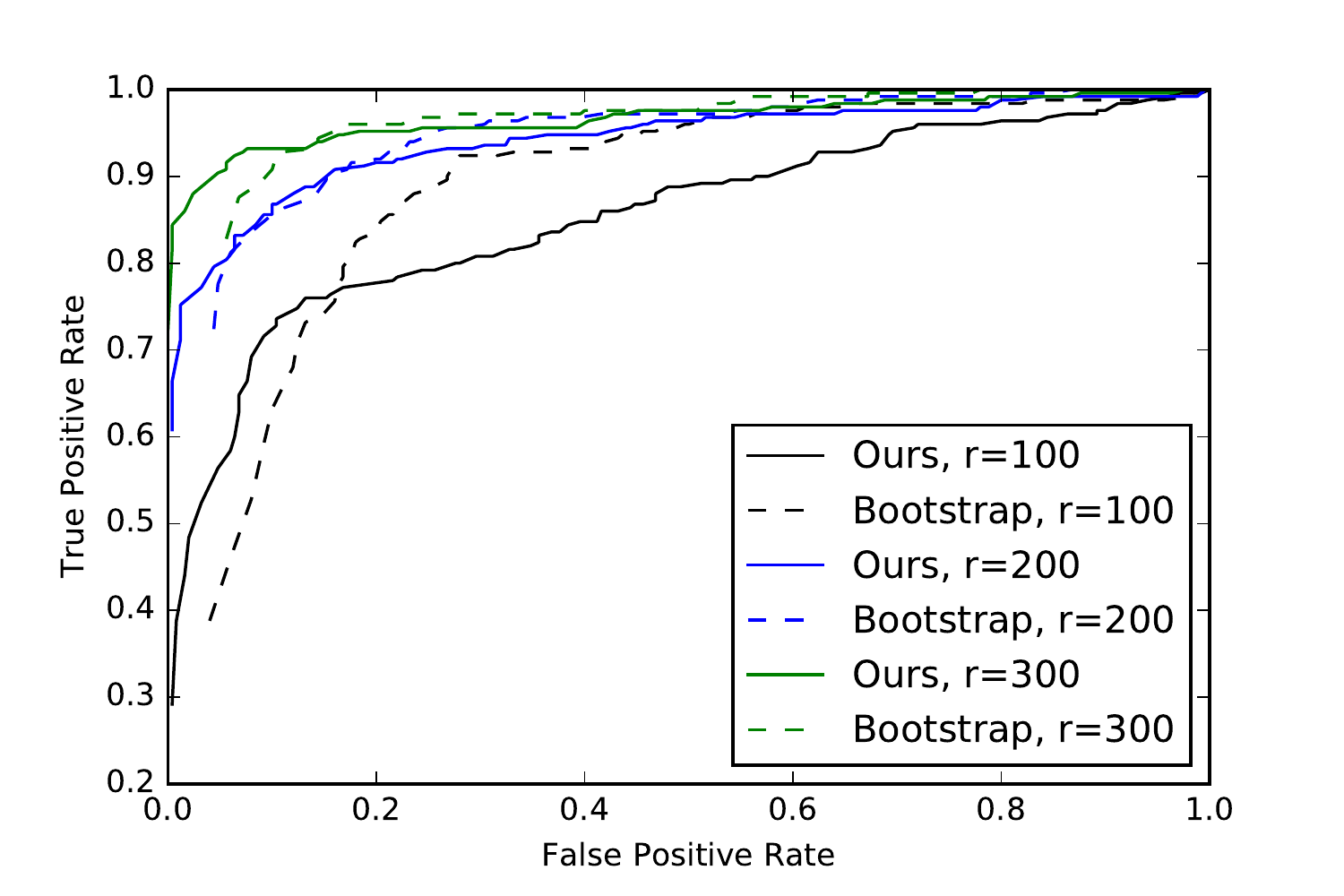}\label{fig:rocp10n5000}
	}
	\subfigure[$p=50$, $N=5000$]{%
		\includegraphics[width=0.45\textwidth]{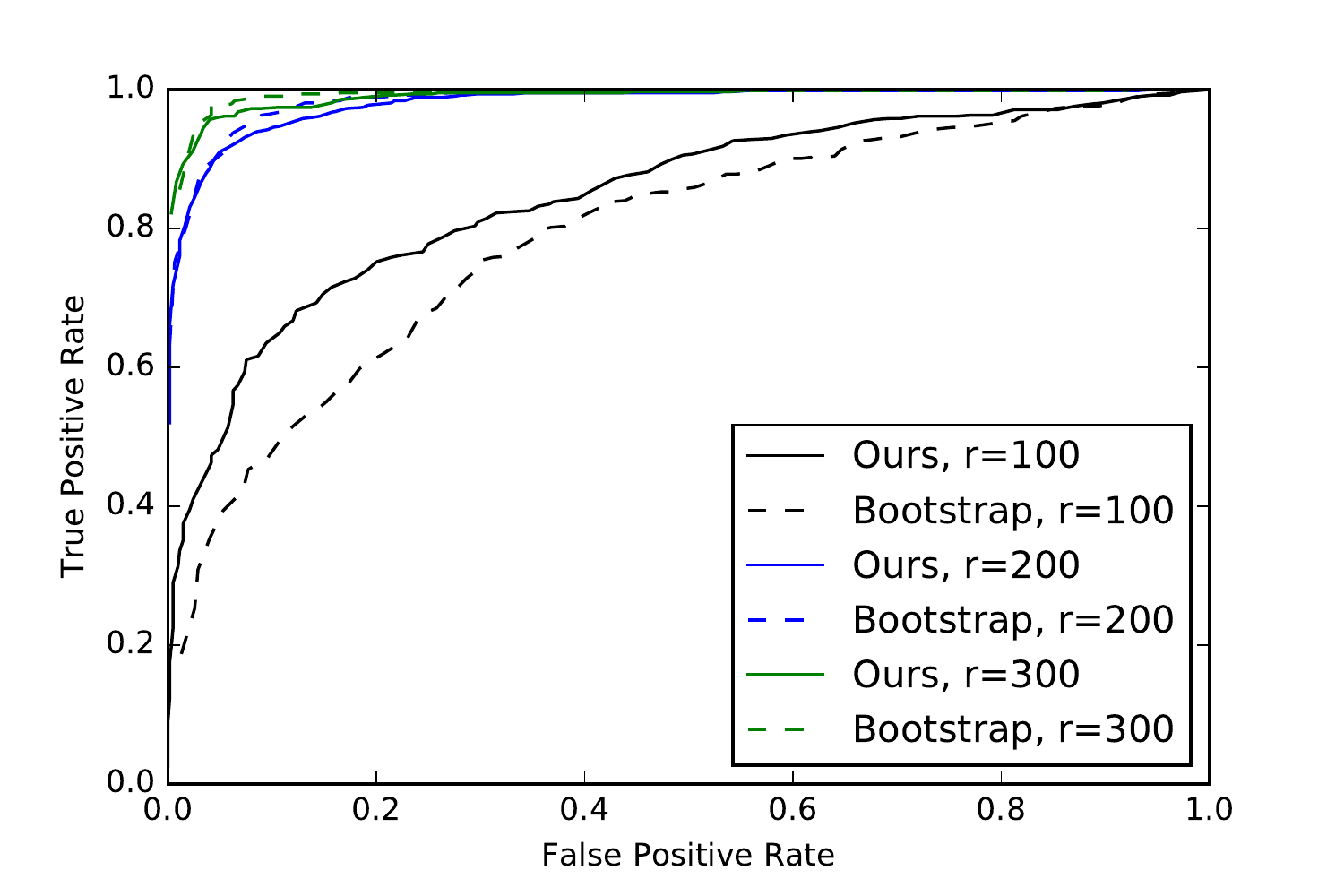}\label{fig:rocp50n5000}
	}
	\caption{ROC curves}
	\label{fig:roc}
\end{figure}

As expected, as $r$ or $N$ increase, the tests of significance improve in terms of the area under the ROC curve. However, there does not seem to be an appreciable difference between the areas under the ROC curve for our algorithm and the bootstrap. Our algorithm seems to have more power when the false positive rate is small and the bootstrap seems to have more power when the false positive rate is large.
Since in practice we would like to limit the false positive rate in order to avoid any costly actions, we believe that our algorithm is superior.

\end{document}